\documentclass[sigconf]{acmart}

\newcommand{\mc}{\mathcal}

\newcommand{\bs}{\boldsymbol}

\usepackage{amsthm}
\usepackage{amssymb}
\usepackage{amsfonts}
\usepackage{amsmath}
\usepackage{eso-pic}
\usepackage{graphicx}
\usepackage{tikz}
\usepackage{pgf}
\usepackage{lscape}
\usepackage{fancyvrb}
\usepackage{verbatim}
\usepackage{xfrac}
\usepackage{algpseudocode,algorithm,algorithmicx}
\usepackage{multirow}
\usepackage{xcolor}
\usepackage{proof}
\usepackage{epstopdf}
\usepackage[font=small]{subfig}
\usepackage{enumitem}
\usepackage{flushend}
\usetikzlibrary{arrows,automata}
\usetikzlibrary{shapes}
\usepackage[font=small]{caption}
\usepackage{balance}
\usepackage{epstopdf}  

\usepackage{xcolor}

\newcommand{\req}{\mathcal{R}}
\newcommand{\reqs}[1]{\req^{\mathcal{#1}}}
\newcommand{\devs}{\mathcal{D}}
\newcommand{\res}{\rho}
\newcommand{\clusters}{\mathcal{K}}
\newcommand{\sm}[1]{\setminus\{{#1}\}}

\newcommand{\reffig}[1]{Fig.~\ref{#1}}   
\newcommand{\reftable}[1]{Table~\ref{#1}}   
\DeclareMathOperator*{\argmax}{argmax} 

\theoremstyle{plain}
\newtheorem{theorem}{Theorem}

\newtheorem*{corollary*}{Corollary}

\theoremstyle{definition}
\newtheorem{definition}{Definition}
\newtheorem*{definition*}{Definition}

\theoremstyle{remark}

\newtheorem*{remark*}{Remark}

\copyrightyear{2019} 
\acmYear{2019} 
\setcopyright{acmcopyright}
\acmConference[MobiHoc '20]{ACM International Symposium on Theory, Algorithmic Foundations, and Protocol Design for Mobile Networks and Mobile Computing}{June 30 -- July 3, 2020}{Shanghai, China}
\acmBooktitle{ACM International Symposium on Theory, Algorithmic Foundations, and Protocol Design for Mobile Networks and Mobile Computing  (MobiHoc '20), June 30 -- July 3, 2020, Shanghai, China}
\acmPrice{15.00}
\acmDOI{10.1145/3323679.3326503}
\acmISBN{978-1-4503-6764-6/19/07}

\title{\textit{Sl-EDGE:} Network Slicing at the Edge}

\author{Salvatore D'Oro$^\dagger$, Leonardo Bonati$^\dagger$, Francesco Restuccia$^\dagger$, Michele Polese$^\dagger$, Michele Zorzi$^\ddagger$ and Tommaso Melodia$^\dagger$}
\affiliation{%
  \institution{$^\dagger$Institute for the Wireless Internet of Things, Northeastern University, Boston, MA, USA \\
  $^\ddagger$Department of Information Engineering, University of Padova, Italy}
}

\renewcommand{\shortauthors}{S.~D'Oro et al.}

\begin{abstract}

Network slicing of multi-access edge computing (MEC) resources is expected to be a \textit{pivotal} technology to the success of 5G networks and beyond. The key challenge that sets MEC slicing apart from traditional resource allocation problems is that edge nodes depend on \textit{tightly-intertwined} and \emph{strictly-constrained} networking, computation and storage resources. Therefore, instantiating MEC slices without incurring in resource over-provisioning  is hardly addressable with existing slicing algorithms. The main innovation of this paper is \textit{Sl-EDGE}, a unified MEC slicing framework that allows network operators to instantiate heterogeneous slice services (\textit{e.g.,} video streaming, caching, 5G network access) on edge devices. We first describe the architecture and operations of \textit{Sl-EDGE}, and then show that the problem of optimally instantiating joint network-MEC slices is NP-hard. Thus, we propose near-optimal algorithms that leverage key similarities among edge nodes and resource virtualization to instantiate heterogeneous slices $7.5$x faster and within $0.25$ of the optimum. We first assess the performance of our algorithms through extensive numerical analysis, and show that \textit{Sl-EDGE} instantiates slices 6x more efficiently then state-of-the-art MEC slicing algorithms. Furthermore, experimental results on a 24-radio testbed with 9 smartphones demonstrate that \textit{Sl-EDGE} provides at once highly-efficient slicing of joint LTE connectivity, video streaming over WiFi, and \textit{ffmpeg} video transcoding. \vspace{-0.2cm}
\end{abstract}



\begin{document}

\maketitle
\pagestyle{plain}


\section{Introduction}

It is now clear that advanced softwarization and virtualization paradigms such as network slicing will be the \textit{cornerstone} of 5G networks~\cite{d2020slicing} and the Internet of Things \cite{bizanis2016sdn}. Indeed, by sharing a common underlying physical infrastructure, network operators~(NOs) can dynamically deploy multiple ``slices'' tailored for specific services (\textit{e.g.,} video streaming, augmented reality), as well as requirements (\textit{e.g.,} low latency, high throughput, low jitter)~\cite{5gkpi}, avoiding the static---thus, inefficient---network deployments that have plagued traditional hardware-based cellular networks. To further decrease latency, increase throughput, and provide improved services to their subscribers, NOs have recently started integrating \textit{multi-access edge computing}~(MEC) technologies~\cite{taleb2017multi}, which are expected to become essential to reach the sub-1\:ms latency requirements of 5G. MEC will be so quintessential that the European Telecommunications Standards Institute~(ETSI) has identified MEC as ``\textit{one of the key pillars for meeting the demanding Key Performance Indicators~(KPIs) of 5G}'' and ``\textit{[as playing] an essential role in the transformation of the telecommunications business}''~\cite{mecjoint}.

Despite the clear advantages of network slicing and MEC, the truth of the matter is that {we cannot have one without the other}. Indeed, slicing networking resources only, \emph{e.g.,} spectrum and resource blocks~(RBs) \emph{cannot suffice to satisfy the stringent timing and performance requirements of 5G networks}.  Real-time wireless video streaming, for example, requires at the same time (i) networking resources (\textit{e.g.,} RBs) to broadcast the video, (ii)~computational resources to process and transcode the video, as well as (iii) storage resources to locally cache the video. The key issue that sets MEC slicing apart from traditional slicing problems is that MEC resources are usually \textit{coupled}, meaning that slicing one resource usually leads to a \textit{performance degradation} in another type of resource.

\begin{figure}[!h]
    \centering
    \includegraphics[width=\columnwidth]{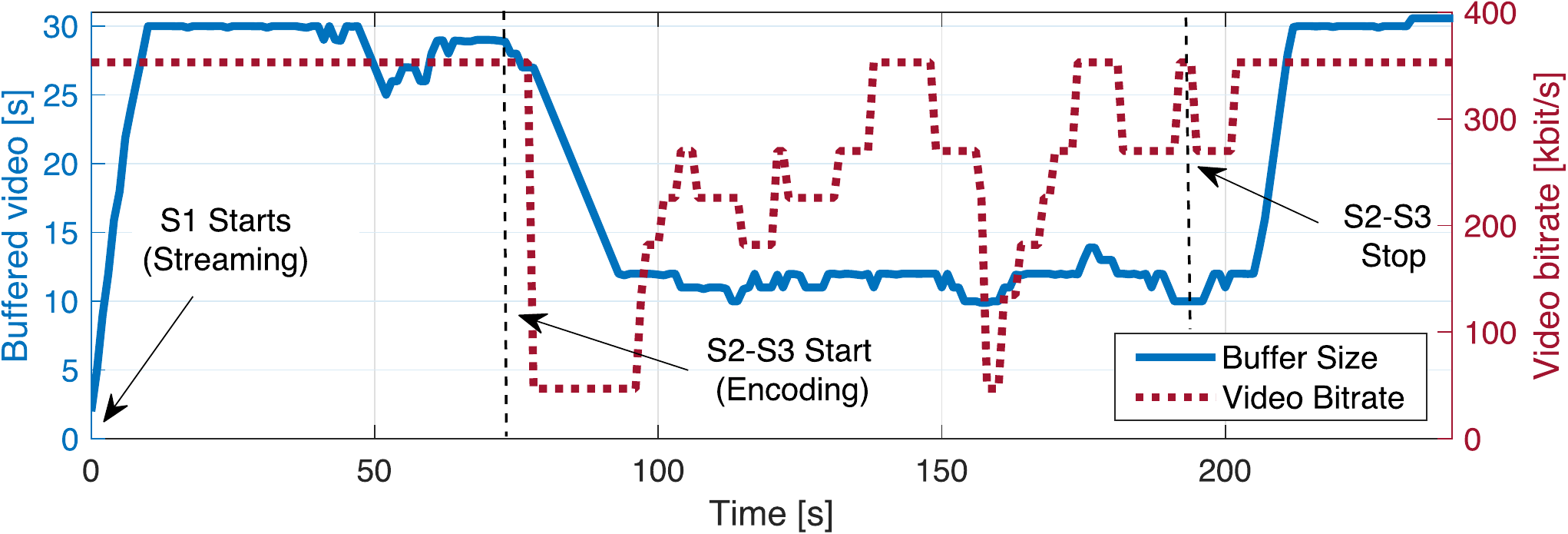}
    \caption{Effect of coupling on joint networking-MEC slicing.}
    \label{fig:coupling}
 \end{figure}

We further verify this critical issue in \reffig{fig:coupling}, where we show an experiment (testbed described in Section \ref{sec:experimental}) where we instantiate 1 slice for video streaming (S1) and 2 slices for video transcoding (S2 and S3). S1 starts at $t = 0$,  while S2 and S3 start at $t = 75$. \reffig{fig:coupling} clearly shows that as soon as S2 and S3 start, the performance of S1 plummets. This is because the computational resources allocated for S2 and S3 cause the video buffer (blue line) to drop from $\sim$30 seconds to $\sim$10 seconds, which in turn causes highly-jittered bitrate (red line). As soon as S2 and S3 end at $t = 190$, buffer size and video bitrate sharply increase and stabilize. This demonstrates that slices that require both computation and networking resources (S1, video streaming) are inevitably impacted by slices running on the same node that only require computation (S2 and S3, video transcoding). Therefore, \textit{taking into account the coupling among slices is a compelling necessity to guarantee appropriate performances when designing edge slicing algorithms}.

Extensive research efforts have already explored MEC algorithms for cellular networks and network slicing \cite{direct,8666109,castellano2019distributed,ndikumana2019joint}. The key limitation of prior work, however, is the assumption that network slicing and MEC are distinct problems. As we demonstrated in \reffig{fig:coupling}, this is hardly the case in practical scenarios. However, the already rich literature on network slicing---discussed in detail in Section~\ref{sec:related}---has not yet considered this fundamental aspect. This makes the joint MEC/slicing paradigm a substantially  unexplored problem. We also point out that due to the massive scale envisioned for 5G and IoT applications, centralized algorithms become prohibitive. Thus, a core issue is not only to account for resource coupling, but also to devise new slicing algorithms that \textit{provide highly efficient and  scalable slicing strategies}.\smallskip

\textbf{Novel contributions.}~The paper's key innovation is the design, analysis and experimental evaluation of the unified MEC slicing framework, called \textit{Sl-EDGE}, that allows network operators to instantiate heterogeneous slice services (\textit{e.g.,} video streaming, caching, 5G network access) on edge devices. In a nutshell, we make the following novel contributions;\smallskip

$\bullet$ We mathematically model and discuss coupling relationships among networking, storage and computation resources at each edge node (Section~\ref{sec:model}). We formulate the \textit{Edge Slicing Problem}~(ESP) as a Mixed Integer Linear Programming~(MILP) problem, and we prove that it is NP-hard (Section~\ref{sec:optimal});\smallskip

$\bullet$ We propose three novel slicing algorithms to address (ESP), each having different optimality and computational complexity. Specifically, we present (i) a centralized optimal algorithm for small network instances (Section~\ref{sec:optimal}); (ii) an approximation algorithm that leverages virtualization concepts to reduce complexity with close-to-optimal performance (Section~\ref{sec:aggreagation}), and (iii) a low-complexity algorithm where slicing decisions are made at the edge nodes with minimal overhead (Section~\ref{sec:distributed});\smallskip

$\bullet$ We extensively evaluate the performance of the three slicing algorithms through simulation, and compare \textit{Sl-EDGE} with DIRECT \cite{direct}, to the best of our knowledge the state-of-the-art slicing framework for MEC 5G applications (Section~\ref{sec:numerical}). Our results show that, by taking into account coupling among heterogeneous resources, \textit{Sl-EDGE} (i) instantiates slices 6x more efficiently then the algorithm proposed in \cite{direct}, as well as satisfying resource availability constraints;
(ii) can be implemented with a distributed approach while getting 0.25 close to the optimal solution;\smallskip

$\bullet$ We prototype and demonstrate \textit{Sl-EDGE} on a testbed composed by 24 software-defined radios. Experimental results demonstrate that \textit{Sl-EDGE} instantiates heterogeneous slices providing LTE connectivity to smartphones, video streaming over WiFi, and \textit{ffmpeg} video transcoding while achieving an instantaneous throughput of $37\:\mathrm{Mbit/s}$ over LTE links, $1.2\:\mathrm{Mbit/s}$ video streaming bitrate with an overall CPU utilization of $83\%$ (Section~\ref{sec:experimental}).




\section{Related Work} \label{sec:related}

Motivated by the ever increasing interests from both NOs and standardization entities \cite{mecslicing, mecdeploy,mecjoint}, network slicing and multi-access edge computing technologies have recently become ``all the rage'' in the wireless research community~\cite{mancuso2019slicing,d2019slice,mandelli2019satisfying,garcia2018posens,rost2017network,zhang2019adaptive,foukas2017orion}. Lately, we have seen a deluge of algorithms to efficiently slice portions of the network and instantiate service-specific slices. These solutions leverage optimization~\cite{han2019utility,halabian2019distributed,sun2019hierarchical,agarwal2018joint}, game-theory~\cite{caballero2017network,jiang2017network,d2018low} and machine learning~\cite{8666109} tools. 

Concurrently, MEC has demonstrated to be an effective methodology to  significantly reduce latency. This paradigm has been successfully used to provide task offloading~\cite{misra2019detour,tran2018joint,liu2016delay,wang2019computation}, augmented reality (AR)~\cite{al2017energy,erol2018caching,tang2018multi}, low-latency video streaming~\cite{he2019cloud,yang2018multi}, and caching~\cite{hoang2018dynamic,zhang2018cooperative}, among others.
We refer the interested reader to the following surveys for a more detailed overview of network slicing and MEC~\cite{afolabi2018network,kaloxylos2018survey,taleb2017multi,mao2017survey,mach2017mobile}.


These works are extremely effective when the two technologies are considered as independent entities operating on the same infrastructure. However, as shown in \reffig{fig:coupling}, they are prone to resource over-provisioning when both technologies coexist on the same edge nodes and share the same resources. Ndikumana \textit{et al.}~\cite{ndikumana2019joint} consider the allocation of heterogeneous resources for task offloading problems in MEC ecosystems, while in~\cite{agarwal2018joint} Agarwal \textit{et al.}\ consider the problem of jointly allocating CPUs and virtual network functions for network slicing applications.
Similarly, in~\cite{8666109} Van Huynh \textit{et al.}\ account for slice networking, computation, and storage resources by designing a deep dueling neural network that determines which slices to be admitted, maximizes the network provider's reward, and avoids resource over-provisioning. Although~\cite{8666109} accounts for resource heterogeneity, the authors only focus on the slicing admission control problem, and do not consider how to partition the resources of each slice among the nodes of the network. 
A heterogeneous resource orchestration framework for edge computing ecosystems called \textit{Senate} is presented by Castellano \textit{et al.}\ in~\cite{castellano2019distributed}. Senate leverages Distributed  Orchestration  Resource  Assignment (DORA) and election algorithms to allow the instantiation of multiple virtual network functions on the same infrastructure. However, this work focuses on the wired portion of the edge network only.  

The closest work to ours is~\cite{direct}, where Liu and Han proposed a distributed slicing framework for MEC-enabled wireless networks called DIRECT. Despite being extremely successful in slicing networks with MEC resources residing in dedicated servers close to the base stations, DIRECT does not account for the case where both networking and MEC resources coexist on the same edge node.

This paper separates itself from prior work as it makes a substantial step forward toward the coexistence of network slicing and MEC technologies. Indeed, we consider the challenging case of edge nodes jointly providing wireless network access and MEC functionalities to mobile users. Furthermore, we model the intrinsic coupling among heterogeneous resources residing on edge nodes, and design \textit{Sl-EDGE}, a slicing framework that leverages such a coupling to instantiate heterogeneous slices on the same physical infrastructure. 

\section{S\lowercase{l}-EDGE at a glance} \label{sec:overview}

\textit{Sl-EDGE} is a slicing framework for MEC-enabled 5G systems. Its key advantage is that it provides a fast, flexible and efficient deployment of joint networking and MEC slices. The three-tiered architecture of \textit{Sl-EDGE} is illustrated in \reffig{fig:architecture}.

\begin{figure}[!t]
    \centering
    \setlength\belowcaptionskip{-.48cm}
    \includegraphics[width=\columnwidth]{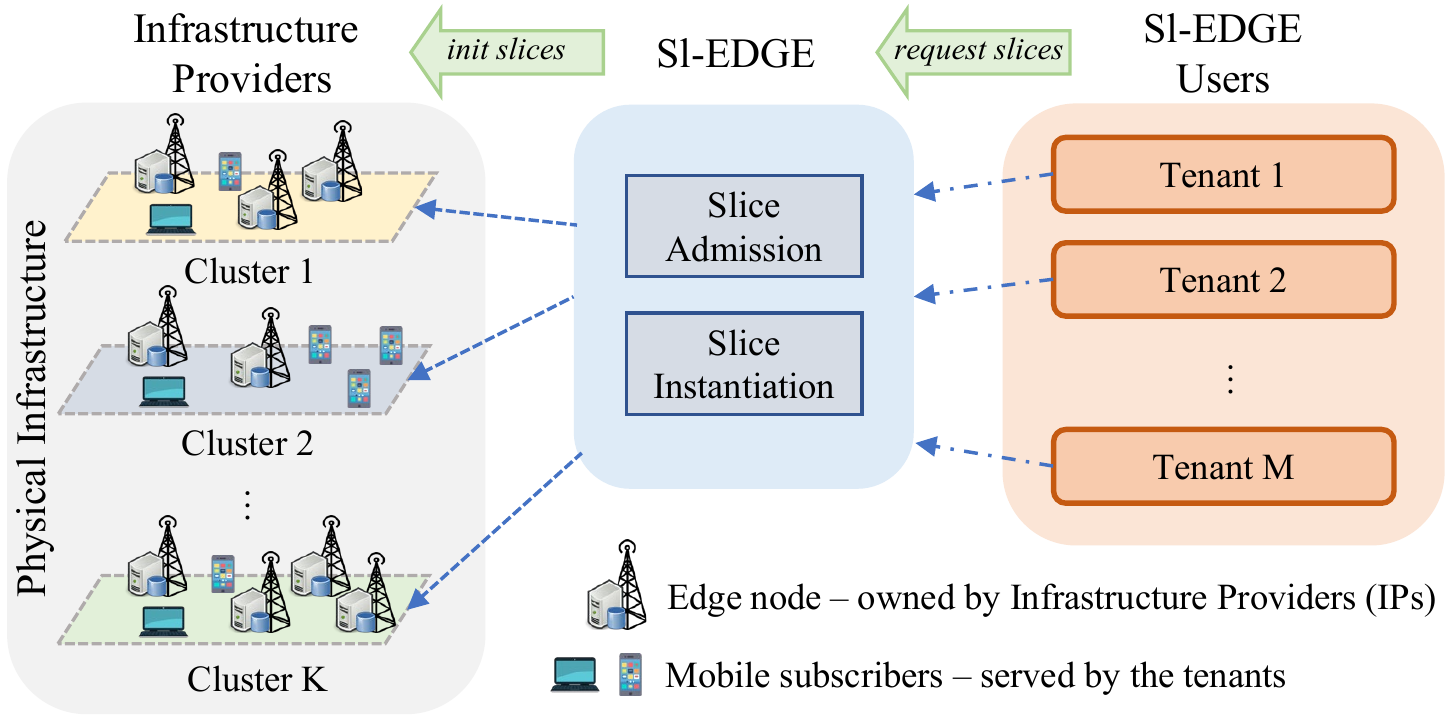}
    \caption{The three-tier architecture of \textit{Sl-EDGE}.}
    \vspace{-0.17cm}
    \label{fig:architecture}
 \end{figure}

The \textit{physical infrastructure} consists of a set of MEC-enabled networking edge nodes (\textit{e.g.,} base stations, access points, IoT gateways)---referred to as \textit{MEC hosts}~\cite{mecjoint}---controlled by one or more Infrastructure Providers~(IPs). MEC hosts are located at the network edge and simultaneously provide networking, storage and computational services (\textit{e.g.},  Internet access, video content delivery, caching).

\textit{Sl-EDGE Users} are both mobile and virtual NOs and Service Providers (SPs)---referred to as the \textit{tenants}---willing to rent portions of the infrastructure to provide services to their subscribers. Tenants access \textit{Sl-EDGE} to visualize relevant information such as position of MEC hosts, which areas they cover and a list of networking and MEC services that can be instantiated on each host (\textit{e.g.,} 5G/WiFi connectivity, caching, computation). Whenever tenants need to provide these services, they submit slice requests to obtain networking, storage or computation resources. The received slice requests are collected and processed by \textit{Sl-EDGE} which (i) determines the set of requests to be accommodated by using centralized (Section~\ref{sec:optimal}) and distributed algorithms (Section~\ref{sec:approximation});
(ii) instantiates slices by allocating the available resources to each admitted slice, and (iii) notifies to the admitted tenants the list of the resources allocated to the slice.

\section{System model} \label{sec:model}
Let $\devs$ be the set of deployed MEC-enabled networking devices, or \textit{edge nodes}.
Edge nodes provide both wireless connectivity and MEC services to a limited portion of the network. Therefore, they can be clustered into $K$ \textit{clusters} located in different geographical areas~\cite{ndikumana2019joint,bouet2018mobile,gudipati2013softran}.
Let $\clusters=\{1,2,\dots,K\}$ be the set of these $K$ independent clusters, and let $\devs_k$ be the set of edge nodes in cluster $k\in\clusters$. Each edge node $d\in\devs_k$ is equipped with a set of networking, storage and computational capabilities, usually measured in terms of number of RBs~\cite{d2019slice}, megabytes, and billion of instructions per second~(GIPS)~\cite{ndikumana2019joint}, respectively.
Let $z\in\mathcal{T}=\{\mc{N},\mc{S},\mc{C}\}$ represent the resource type, \textit{i.e.},  networking ($\mc{N}$), storage ($\mc{S}$), and computing ($\mc{C}$). Moreover, let $\boldsymbol{\res}_d=(\res_d^z)_{z\in\mathcal{T}}\in\mathbb{R}_{\geq0}^3$ be the set of resources available at each edge node $d$. 
An example of the physical infrastructure and its clustered structure is shown in \reffig{fig:example_network}.

Let $\req =\{\reqs{N},\reqs{S},\reqs{C}\}$ be the set of slice requests submitted to \textit{Sl-EDGE}, with $\reqs{N},\reqs{S},\reqs{C}$ being the set for networking, storage, and computing slice requests, respectively. Each request $r\in\req^z$ of type $z$ is associated to a \textit{value} $v_r^z > 0$  used by the IP to assess the importance---or monetary value---of $r$. Also, we define the $K$-dimensional \textit{request demand array} $\boldsymbol{\tau}_r=(\tau^z_{r,k})_{k\in\clusters}$, where $\tau^z_{r,k}\geq0$ represents the amount of resources of type $z$ requested by $r$ in cluster $k$. Without loss of generality, we assume that $\sum_{k\in\clusters} \tau^z_{r,k} > 0$ for all $r\in\req$.

\begin{figure}[!t]
    \centering
    \includegraphics[width=\columnwidth]{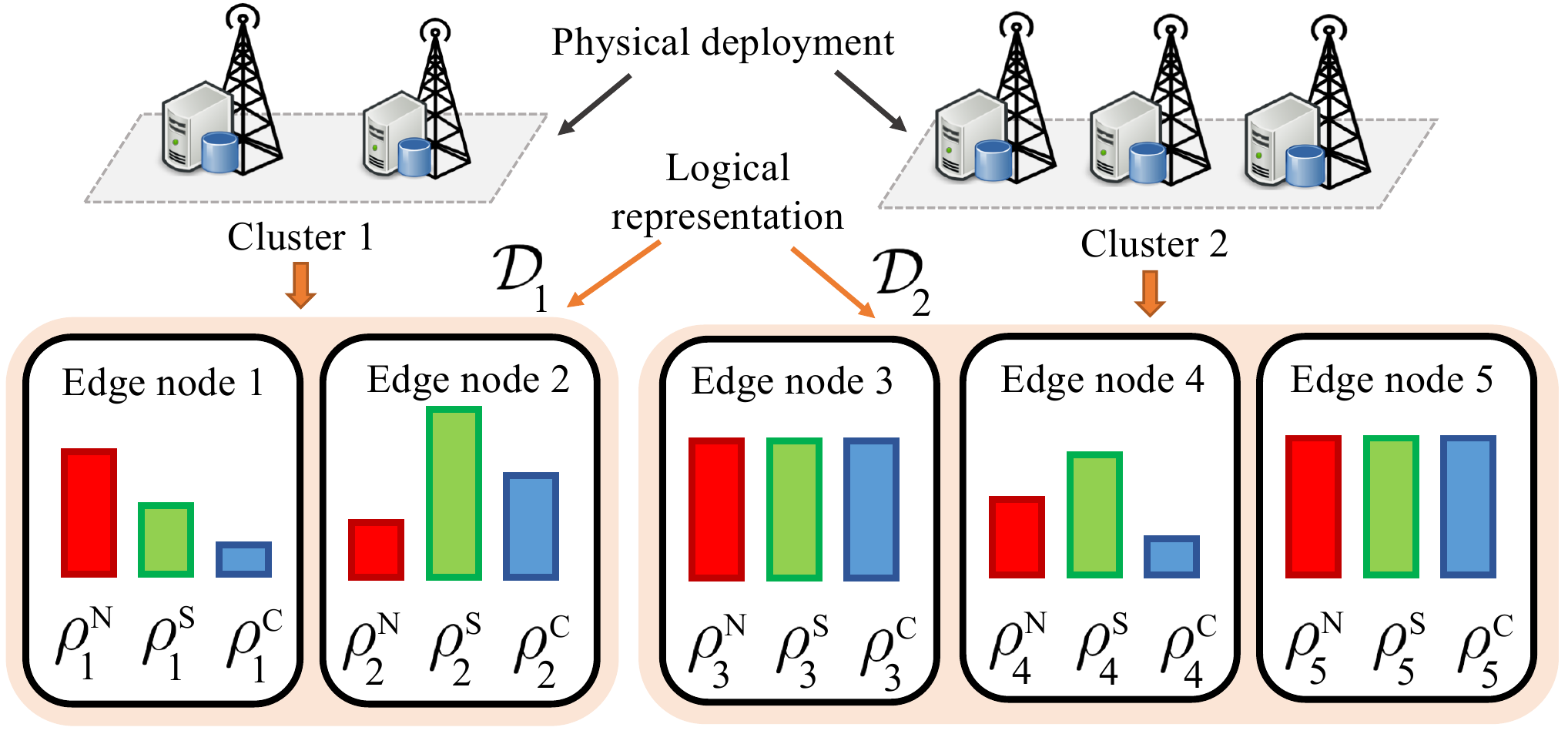}
    \caption{System model example with $K=2$ clusters with edge node sets $\devs_1=\{1,2\}$ and $\devs_2=\{3,4,5\}$, respectively.}
    \label{fig:example_network}
    \vspace{-0.15cm}
 \end{figure}

\subsection{Resource coupling and collateral functions} \label{sec:model:coupling}

To successfully slice networking and MEC resources, it is paramount to understand the underlying dynamics  between resources of different nature. To this purpose, let us consider two simple but extremely effective examples.

\begin{figure}[h!]
\centering
   \begin{minipage}[b]{0.49\columnwidth}
    \centering
    \includegraphics[width=\columnwidth]{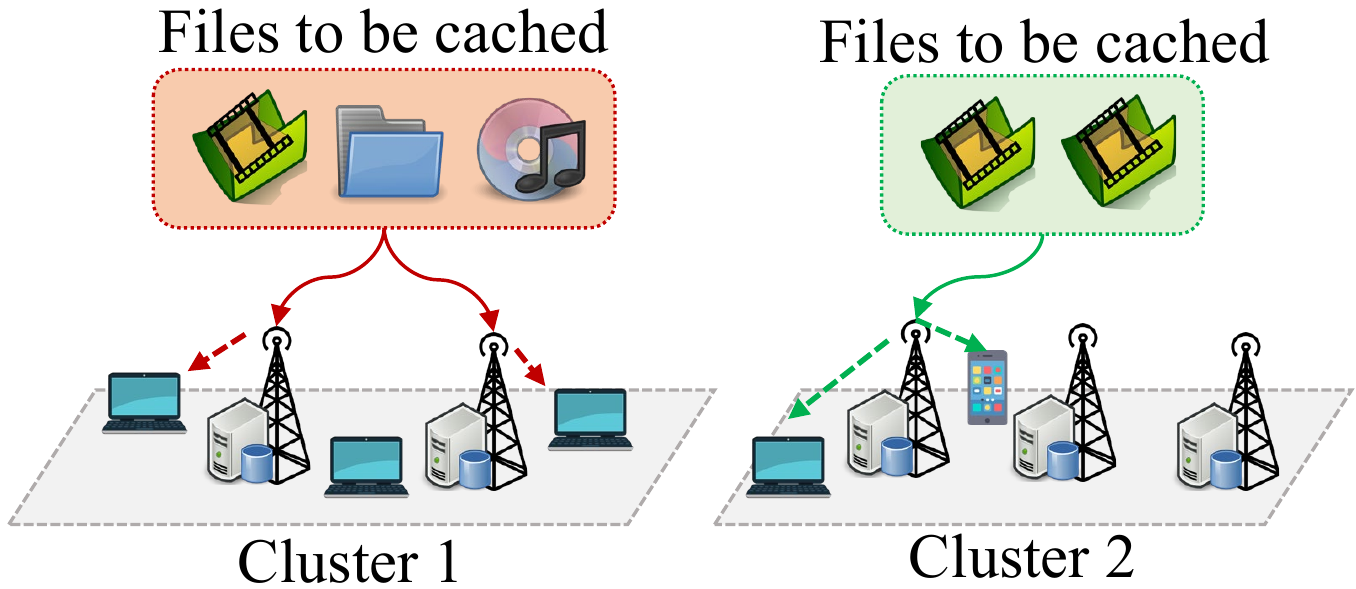}
    \caption{\label{fig:coupling:caching} Content caching.}
  \end{minipage}
  \begin{minipage}[b]{0.49\columnwidth}
    \includegraphics[width=\columnwidth]{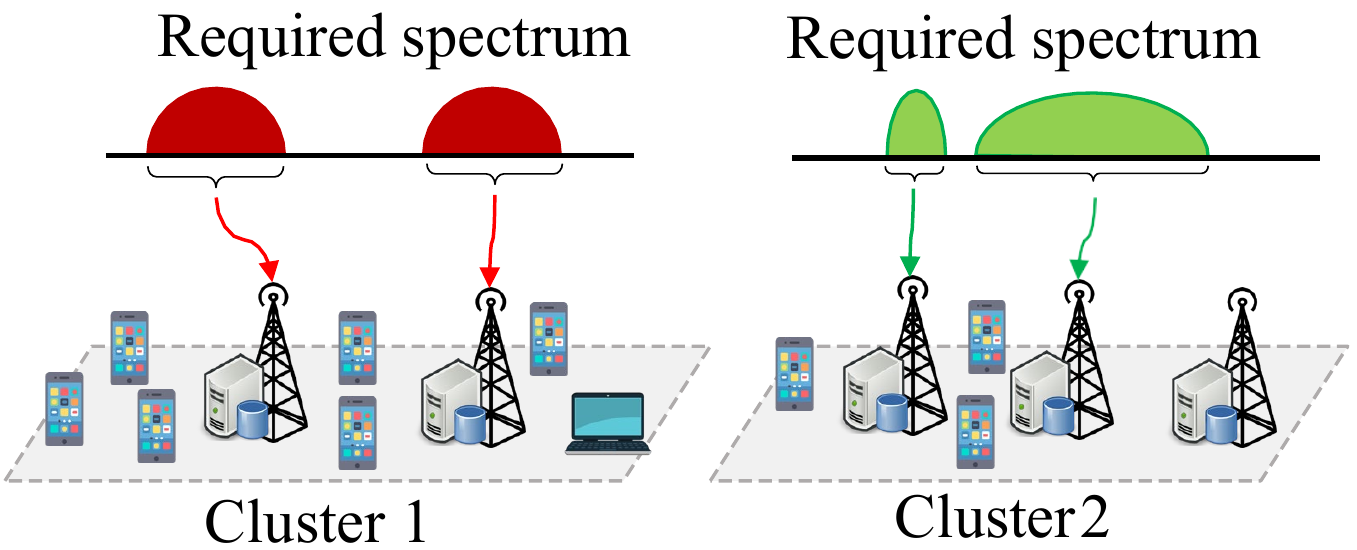}
    \caption{\label{fig:coupling:network} 5G networking.}
  \end{minipage}
\end{figure}

\subsubsection{Content Caching} a tenant instantiates a storage slice (\reffig{fig:coupling:caching}) to provide caching services to its subscribers, \textit{i.e.,} $r\in\req^{\mc{S}}$, and specifies how many megabytes ($\tau_{r,k}^{\mc{S}}$) should be allocated in each cluster $k\in\clusters$.
In this case, the content to be cached should be (i) first transmitted and then (ii) processed by storing edge nodes.
Therefore, storage activities related to caching procedures not only utilize storage resources, but also require networking and computational resources.
\subsubsection{5G networking} in this example (\reffig{fig:coupling:network}) a tenant wants to provide cellular services to mobile subscribers. Hence, it submits a networking slice request $r$ of type $\mc{N}$ and specifies the clusters to be included in the slice and the amount of spectrum resources ($\tau_{r,k}^{\mc{N}}$) needed in each cluster. Edge nodes providing connectivity must (i) perform channel estimation and baseband signal processing procedures, and (ii) locally cache or buffer the data to transmit. Therefore, the allocation of resources of type $\mc{N}$ entails resources of type $\mc{C}$ and $\mc{S}$.

These two examples suggest that heterogeneous resources are tightly intertwined, thus motivating the need for new slicing algorithms that account for these intrinsic relationships. To incorporate coupling within the \textit{Sl-EDGE} framework, we introduce the concept of \textit{collateral functions}.

Let us consider the case where, to instantiate a slice of type $z\in\mc{T}$, $x$ resources must be allocated on edge node $d \in \devs_k, k \in \clusters$. For any resource type $t\in\mathcal{T}\setminus\{z\}$, we define the collateral function $\alpha_{d,k}^{z\rightarrow t}(x):\mathbb{R}\rightarrow\mathbb{R}$. This function (i) reflects coupling among heterogeneous resources, and (ii) determines how many resources of type $t$ should be allocated on edge node $d$ when allocating $x$ resources of type $z$. Of course, $\alpha_{d,k}^{z\rightarrow z}(x)=x$.

In this paper, we model resource coupling as an increasing linear function with respect to $x$. This way, the number of resources of type $z$ needed to instantiate $x$ resources of type $t$ on a given edge node $d$ can be evaluated as $\alpha_{d,k}^{t\rightarrow z}(x)=A_{d,k}^{t\rightarrow z}x$, with $A_{d,k}^{t\rightarrow z}$ being measured in units of type $z$ per unit of type $t$, \textit{e.g.}, GIPS per megabyte.\footnote{We assume that $A_{d,k}^{t\rightarrow z}$ differs among edge nodes, but it is uniform across services of type $t$. When different services of type $t$ have different values of $A_{d,k}^{t\rightarrow z}$ (\textit{e.g.,} video encoding and file compression might require a different number of GIPS to process the same data) we can extend $\mc{T}$ by adding service-specific classes with different $A_{d,k}^{t\rightarrow z}$ values.} 
However, the more general case where $\alpha_{d,k}^{t\rightarrow z}(x)$ is a non-linear function can be easily related to the linear case by using well-established and accurate piece-wise linearization techniques~\cite{lin2013review}. For any $k\in\clusters$ and $d\in\devs_k$, let $\mathbf{A}_{d,k}$ be the \textit{collateral matrix} for edge node $d$. Such a matrix can be written as 

\smallskip
\begin{equation} \label{mat:alpha}
    \mathbf{A}_{d,k}= \left(
    \begin{matrix}
     1                                  & A_{d,k}^{\mc{S}\rightarrow \mc{N}} & A_{d,k}^{\mc{C}\rightarrow \mc{N}}\\
     A_{d,k}^{\mc{N}\rightarrow \mc{S}} & 1                                  & A_{d,k}^{\mc{C}\rightarrow \mc{S}}\\
     A_{d,k}^{\mc{N}\rightarrow \mc{C}} & A_{d,k}^{\mc{S}\rightarrow \mc{C}} & 1
    \end{matrix}
    \right).
\end{equation}
\smallskip

\section{Edge slicing problem and its optimal solution} \label{sec:optimal}
\smallskip
\smallskip
\smallskip

The key targets of \textit{Sl-EDGE} are to (i) maximize profits generated by infrastructure slice rentals, and (ii) allow location-aware and dynamic instantiation of slices in multiple clusters, while (iii) avoiding over-provisioning of resources to avoid congestion and poor performance. We formalize the above three targets with the edge slicing optimization problem~(ESP) introduced below. 

\begin{align} 
 \underset{\mathbf{y},\bs{\sigma}}{\text{maximize}} &  \sum_{k\in\clusters}\sum_{z\in\mc{T}} \sum_{r\in\req^z} v^z_{r} y^z_{r} \label{prob:utility} \tag{ESP} \\
 \text{subject to} & \hspace{-0.1cm}\sum_{d\in\devs_k} \sigma_{r,d}^z = \tau_{r,k}^z y_{r}^z, \hspace{0.05cm} \forall z\in\mc{T}, k\in\clusters, r\in\req^{z} \label{eq:c1} \\
& \hspace{-0.25cm} \sum_{r\in\req^z}\! \sum_{t\in\mc{T}}  \! \! \alpha_{d,k}^{t\rightarrow z}(\sigma_{r,d}^t) \!
\leq \! \res_{d,k}^z, \hspace{0.05cm} \forall \! z\! \in\! \mc{T}, k\! \in\! \clusters,  d\! \in\! \devs_k \label{eq:c2} \\
& \hspace{-0.1cm} y_r^z\in\{0,1\}, 
\hspace{0.2cm} \forall z\in\mc{T}, r\in\req^z \\
& \hspace{-0.1cm} \sigma_{r,d}^z \geq 0, 
\hspace{0.1cm} \forall z\in\mc{T}, r\in\req^z, k\in\clusters, d\in\devs_k \label{eq:c4}
\end{align}
\noindent
\smallskip
where $\mathbf{y}=(y_r^z)_{z\in\mc{T},r\in\req^z}$ and $\bs{\sigma}=(\sigma_{r,d}^z)_{z\in\mc{T},r\in\req^z,d\in\devs}$ respectively are the \textit{slice admission} and \textit{resource slicing} policies. Quantity $y_r^z$ is a binary variable such that $y_r^z=1$ if request $r$ is admitted, $y_r^z=0$ otherwise.
Similarly, $\sigma_{r,d}^z$ represents the amount of resources of type $z$ that are assigned to $r$ on edge node $d$.

One can easily verify that Problem~\eqref{prob:utility} meets the previously mentioned requirements, since it (i) aims at maximizing the total value of the admitted slice requests; (ii) guarantees that each admitted slice obtains the required amount of resources in each cluster (Constraint~\eqref{eq:c1}), and (iii) prevents resource over-provisioning on each edge node (Constraint~\eqref{eq:c2}).

Given the presence of both continuous and 0-1 variables, \eqref{prob:utility} belongs to the class of 
MILPs problems, well-known to be hard to solve. More precisely, in Theorem~\ref{th:hard} we will prove that \eqref{prob:utility} is NP-hard even in the case of requests having the same value and edge nodes belonging to a single cluster.

\begin{theorem}[NP-hardness] \label{th:hard}
Problem \eqref{prob:utility} is NP-hard.
\end{theorem}
\begin{proof}
To prove this theorem, we reduce the Splittable Multiple Knapsack Problem (SMKP), which is NP-hard~\cite{10.1007/978-3-319-18173-8_27}, to an instance of \eqref{prob:utility}. Let us assume that all edge nodes belong to the same cluster $k$ and all submitted slice requests are of the same type $z\in\mc{T}$. Furthermore, we assume that all requests have value $v_r^z=v_s^z=1$ for any $(r,s)\in\req\times\req$. 
Since all requests are of the same type $z$, $\alpha_{d,k}^{z\rightarrow z}(x)=x$ for any edge node $d\in\devs$.
Let us now consider the SMKP, whose statement is as follows: \textit{given a set of knapsacks (the edge nodes) with limited capacity ($\res^z_{d,k}$) and a set of items (requests) with certain value ($v_r^z$) and size ($\tau^z_{r,k}$), assume that items can be split among the knapsacks while satisfying Constraint~\eqref{eq:c1}, is there an allocation policy that maximizes the total number of items added to the knapsacks without overfilling them?} We observe that Problem \eqref{prob:utility} is a reduction of the SMKP. Since this reduction can be built in polynomial time, it follows that Problem~\eqref{prob:utility} is NP-hard. 
\end{proof}

Problem \eqref{prob:utility} can be solved by means of efficient and well-established exact Branch-and-Cut (B\&C) algorithms. Even though the worst-case complexity of such algorithms is exponential, B\&C leverages structural properties of the problem to confine the search space, thus reducing the time needed to compute an optimal solution.
The B\&C procedure---not reported here for the sake of space---can be found in~\cite{Elf2001}. We now focus on how to overcome some of the limitations of B\&C. Specifically, B\&C suffers from high computational complexity, and  requires a centralized entity with perfect knowledge, both of which are unacceptable in large-scale and dynamic networks. 
%


\section{Approximation Algorithms} \label{sec:approximation}

We now design two approximation algorithms for \eqref{prob:utility} whose primary objective is to (i) reduce the computational complexity of the problem, and (ii) minimize the overhead traffic traversing the network.
In Section~\ref{sec:aggreagation} and Section~\ref{sec:distributed}, we present the algorithmic implementation of the two algorithms, and further discuss their
optimality, complexity and overhead.

\subsection{Decentralization through virtualization} \label{sec:aggreagation}

One of the main sources of complexity in Problem~\eqref{prob:utility} is the large number of optimization variables $\bf{y}$ and $\bs{\sigma}$. However, we notice that $R = \sum_{z\in\mc{T}}|\req^z|$, where $|\cdot|$ is the set cardinality operator. On the contrary, the number of $\bs{\sigma}$ variables is $\mc{O}(RD)$, with $D$ being the total number of edge nodes in the infrastructure. While $R$ is generally limited to a few tens of requests, the number $D$ of edge nodes deployed in the network might be very large. However, a big portion of these edge nodes are equipped with hardware and software components that are either similar or exactly the same.  Thus, we can leverage similarities among edge nodes to reduce the complexity of Problem \eqref{prob:utility} while achieving close-to-optimal solutions and reduced control overhead.

Edge nodes with similar collateral functions will behave similarly.
However, being similar in terms of $\alpha$ only does not suffice to determine whether or not two edge nodes are similar. In fact, nodes with similar $\alpha$ might have a different amount of available resources. For this reason, we leverage the concept of \textit{similarity functions}~\cite{huang2008similarity}.
\begin{definition} \label{def:similar}
Let $\Delta(d',d''):\devs\times\devs\rightarrow\mathbb{R}$ be a function to score the similarity between edge nodes $d'$ and $d''$. 
Two edge nodes $d',d''\in\devs$ are said to be $\epsilon$\textit{-similar} with respect to $f$ if $\Delta(d',d'')\leq \epsilon$, for any $\epsilon\in\mathbb{R}_{\geq0}$. 
If $\Delta(d',d'')=0$, we say that $d'$ and $d''$ are identical.
\end{definition}

Through
$\epsilon$-similarity, we can first determine which edge nodes inside the same cluster are similar, and then abstract their physical properties to generate a \textit{virtual edge node}. 
For the sake of generality, in this paper we do not make any assumption on $\Delta(\cdot)$ (interested readers are referred to~\cite{xu2005survey} for an exhaustive survey on the topic).
However, the impact of $\Delta(\cdot)$ and $\epsilon$ on the overall system performance will be first discussed in Section~\ref{sec:approximation:impl}, and then evaluated in Section~\ref{sec:numerical}.

\begin{figure}[!t]
    \centering
    \setlength\belowcaptionskip{-.3cm}
    \includegraphics[width=\columnwidth]{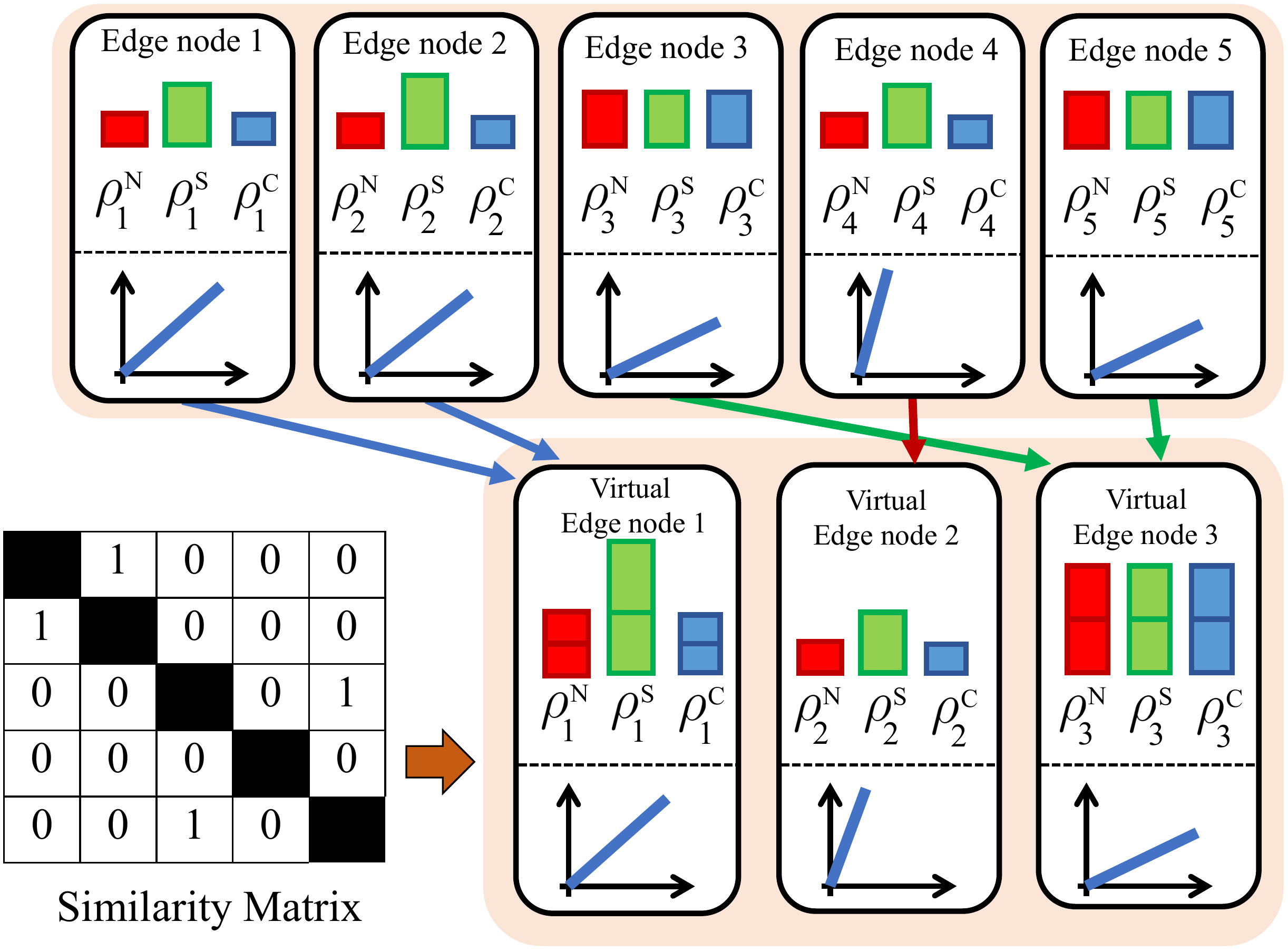}
    \caption{An example of the virtual edge node generation in Step 1. The similarity matrix determines which edge nodes can be aggregated. Similar edge nodes (\textit{i.e.}, $\{1,2\}$ and $\{3,5\}$) are aggregated into virtual ones. Edge node 4 is not aggregated as it has similar resources to $\{1,2\}$, but different collateral function.}
    \vspace{-0.2cm}
    \label{fig:virtualization}
 \end{figure}

Now we present V-ESP, an approximation algorithm that leverages virtualization concepts to compute a solution to Problem \eqref{prob:utility}. The main steps of V-ESP are as follows:

    $\bullet$ \textit{Step~1: (Virtual edge node generation):} for each cluster $k$, we build the $|\devs_k|\times|\devs_k|$ \textit{similarity matrix} $\mc{S}_k$. For any real $\epsilon\geq0$, element $s_{d',d''}\in \mc{S}_k$ indicates whether or not $d'$ and $d''$ are $\epsilon$-similar. That is, $s_{d',d''}=1$ if $\Delta(d',d'')\leq\epsilon$, $s_{d',d''}=0$ otherwise.
    We partition the set $\devs_k$ into $G_k\geq1$ independent subsets that contain similar edge nodes only. Partitions are generated such that $\bigcup_{g=1}^{G_k} \devs_{k,g} = \devs_k$ and $\devs_{k,j} \cap \devs_{k,i} = \emptyset$ for any $i,j=1,2,\dots,G_k$.
 
     Each non-singleton partition is converted into a virtual edge node. Specifically, for each non-singleton partition $\devs_{k,g}$, we define a virtual edge node $\tilde{d}_g$ whose available resources are equal to the sum of the available resources of all edge nodes in the partition, \textit{i.e.,} $\res^z_{\tilde{d}_g,k}=\sum_{d\in\devs_{k,g}} \res^z_{d,k}$. The collateral function of the virtual edge node $d_g^v$ is constructed as
        $\alpha_{\tilde{d}_g,k}^{t\rightarrow z} = f(\devs_{k,g},t,z)$,
    where $f(\cdot)$ is a function that generates a virtualized collateral function for virtual edge node $\tilde{d}_g$ discussed in Section~\ref{sec:approximation:impl}. We show an example of virtualization procedure in \reffig{fig:virtualization}.

    $\bullet$ \textit{Step~2: (Virtual Edge Nodes Advertisement):} each cluster $k$ advertises \textit{Sl-EDGE} the set $\devs_k=(\tilde{d}_g)_{g=1,\dots,G_k}$ of $G_k$ virtualized edge nodes, as well as their virtual collateral functions ($\alpha_{\tilde{d}_g,k}^{t\rightarrow z}$) and available resources ($\res^z_{\tilde{d}_g,k}$).

    $\bullet$ \textit{Step~3: (Solve virtualized ESP):} \textit{Sl-EDGE} solves Problem \eqref{prob:utility} with virtualized edge nodes through B\&C. Slice admission and resource slicing policies $(\bf{\tilde{y}}^*,\bs{\tilde{\sigma}}^*)$ are computed and each cluster receives the $2$-tuple $(\bf{\tilde{y}}^*,\bs{\tilde{\sigma}}^*_k)$, with $\bs{\tilde{\sigma}}^*_k = (\tilde{\sigma}^{z*}_{r,\tilde{d}_g})_{z\in\mc{T},r\in\req,g=1,\dots,G_k}$ being the resource allocation policy over the virtualized edge nodes of cluster $k$.

    $\bullet$ \textit{Step~4: (Virtualized edge node resource allocation):} upon receiving $(\bf{\tilde{y}}^*,\bs{\tilde{\sigma}}^*_k)$, cluster $k$ solves $G_k$ Linear Programming (LP) problems in parallel, one for each virtual edge node $g$. These LPs are formulated as follows:
    \begin{align} 
     \text{find} & \hspace{0.3cm} \bs{\sigma}_{k,g} \label{prob:approx:utility}  \\
     \text{subject to} & \hspace{0.3cm} \sum_{d\in\devs_{k,g}} \sigma_{r,d}^z = \tilde{\sigma}^{z*}_{r,\tilde{d}_g}, \hspace{0.1cm} \forall r\in\req  \\ 
    & \hspace{0.1cm} \mbox{Constraints}~\eqref{eq:c2}, ~\eqref{eq:c4} \nonumber
    \end{align}
    \noindent
    which can be optimally solved by computing any feasible resource allocation policy that satisfies all constraints.

    $\bullet$ \textit{Step~5: (Slicing Policies Construction):} let $\bs{\sigma}^*_{k,g}=(\bs{\sigma}^*_{k,d})_{d\in\devs_{k,g}}$ be the optimal solution of the $g$-th instance of~\eqref{prob:approx:utility}. The resource slicing policy $\bs{\sigma}^*_{k}$ for cluster $k$ is constructed by stacking all $G_k$ individual solutions computed by individual clusters, \textit{i.e.}, $\bs{\sigma}^*_{k} = (\bs{\sigma}^*_{k,g})_{g=1,\dots,G_k}$. The final slice admission and resource slicing policies are $(\bf{\tilde{y}}^*,\bs{\sigma}^*)$ with $\bs{\sigma}^*=(\bs{\sigma}^*_{k})_{k\in\clusters}$.

\smallskip
Through V-ESP, each cluster exposes $G_k\leq|\devs_k|$ virtual edge nodes only, rather than $|\devs_k|$ (Steps~1-2). Thus, virtualization reduces the number of edge nodes and thus the number of variables in~\eqref{prob:utility}.  Moreover, since virtualization leaves the structure of the slicing problem unchanged, we  efficiently solve Step~3 through the same B\&C techniques used for~\eqref{prob:utility}. In addition, while Steps~3-5 are executed whenever a new slicing policy is required (\textit{e.g.,} tenants submit new slice requests or the slice rental period expires), Steps~1-2 are executed only when the structure of the physical infrastructure changes (\textit{e.g.,} edge nodes are turned on/off or are subject to hardware modifications). This way, we can further reduce the overhead.
In short, V-ESP splits the computational burden among the NO (Step 3) and the edge nodes (Steps~1-2 and~4), which jointly provides the high efficiency typical of centralized approaches while enjoying reduced complexity of decentralized algorithms. 

\subsubsection{Design Aspects of Virtualization} \label{sec:approximation:impl}


Step~1 relies on $\epsilon$-similarity to aggregate edge nodes and reduce the search space. Intuitively, the higher the value of $\epsilon$, the smaller the set of virtual edge nodes generated in Step 1, the faster \textit{Sl-EDGE} computes solutions in Step~3. However, large $\epsilon$ values might group together edge nodes with different available resources and collateral functions. In this case, (i) Step~1 might produce virtual edge nodes that poorly reflect physical edge nodes features, and (ii) solutions computed at Step~3 might not be feasible when applied to Step~4. Thus, there is a trade-off between accuracy and computational speed, which will be the focus of Section~\ref{sec:numerical:epsilon}. 

Another aspect that influences the efficiency and feasibility of solutions generated by the V-ESP algorithm is the function $f(\cdot)$, which transforms collateral functions of similar edge nodes into an aggregated collateral function.
Recall that $f(\cdot)$, which can be represented as a collateral matrix~\eqref{mat:alpha}, must mimic the actual behavior of physical edge nodes belonging to the same partition $g$. 
To avoid overestimating the capabilities of virtual edge nodes, and producing unfeasible solutions, the generic element of the virtual collateral matrix~\eqref{mat:alpha} for virtual edge node $\tilde{d}_g$ is set to $A_{\tilde{d}_g,k}^{z\rightarrow t} = \max_{d\in\devs_{k,g}} \{A_{d,k}^{z\rightarrow t}\}, \forall z, t \in \mc{T}$. Although this model underestimates the capabilities of physical edge nodes and may admit less requests than the optimal algorithm, it always produces feasible solutions in Steps~3 and~4.

\subsection{Distributed Edge Slicing} \label{sec:distributed}

In this section we design a distributed edge slicing algorithm for Problem~\eqref{prob:utility} such that clusters can locally compute slicing strategies. We point out that making~\eqref{prob:utility} distributed is significantly challenging. In fact, both utility function and constraints are coupled with each other through the optimization variables $\bs{\sigma}$ and $\mathbf{y}$. This complicates the decomposition of the problem into multiple independent sub-problems. 

In order to decouple the problem into multiple independent sub-problems, we introduce the auxiliary variables $\mathbf{y}_k=(y^z_{r,k})_{z\in\mc{T}, r\in\req^z}$ such that $y^z_{r,k} = y^z_r$ for any request $r$ and cluster $k$. 
Thus, Problem \eqref{prob:utility} can be rewritten as
\begin{align} 
 \underset{\bs{\sigma},\mathbf{y}}{\text{maximize}} & \hspace{0.2cm} \frac{1}{|\clusters|} \sum_{k\in\clusters}\sum_{z\in\mc{T}} \sum_{r\in\req^z} v_r^z y^z_{r,k} \label{prob:desp:utility} \tag{D-ESP} \\
 \text{subject to} & \hspace{0.1cm} \sum_{d\in\devs_k} \sigma_{r,d}^z = \tau_{r,k}^z y_{r,k}^z, \hspace{0.1cm} \forall z\in\mc{T}, k\in\clusters, r\in\req^{z} \label{eq:desp:c1} \\ 
& \hspace{0.2cm} y_{r,k}^z = y_{r,m}^z, \hspace{0.2cm} \forall z\in\mc{T}, (k,m)\in\clusters^2,  r\in\req^z
\label{eq:desp:c2} \\
& \hspace{0.2cm} y_{r,k}^z\in\{0,1\} 
\hspace{0.2cm} \forall z\in\mc{T}, r\in\req^z \\
& \hspace{0.2cm} \mbox{Constraints}~\eqref{eq:c2},~\eqref{eq:c4} \nonumber
 \vspace{-0.7cm}
\end{align}
where $\mathbf{y}=(y_{r,k}^z)_{z\in\mc{T}, r\in\req^z,z\in\mc{T}}$, while Constraint \eqref{eq:desp:c2} guarantees that different clusters admit the same requests.

Problem \eqref{prob:desp:utility} is with separable variables with respect to the $K$ clusters.
That is, Problem~\eqref{prob:desp:utility} can be split into $K$ sub-problems, each of them involving only variables controlled by a single cluster. To effectively decompose Problem~\eqref{prob:desp:utility} we leverage the Alternating Direction Method of Multipliers~(ADMM) \cite{boyd2011distributed}. The ADMM is a well-established optimization tool that enforces constraints through quadratic penalty terms and generates multiple sub-problems that can be iteratively solved in a distributed fashion. 

The augmented Lagrangian for Problem \eqref{prob:desp:utility} can be written as follows:
\begin{align}
    L(\bs{\sigma},\mathbf{y},\bs{\lambda},\rho) & = \sum_{k\in\clusters}\sum_{z\in\mc{T}} \sum_{r\in\req^z} v_r^z y^z_{r,k} \nonumber \\
    & - \sum_{z\in\mc{T}} \sum_{r\in\req^z}  \sum_{k\in\clusters}  \sum_{m\in\clusters} \lambda^z_{r,k,m} (y^z_{r,k}-y^z_{r,m}) \nonumber \\
    & - \frac{\rho}{2} \sum_{z\in\mc{T}} \sum_{r\in\req^z}  \sum_{k\in\clusters}  \sum_{m\in\clusters}  (y^z_{r,k}-y^z_{r,m})^2 \label{eq:lagrangian}
\end{align}
\noindent
where $\bs{\lambda}=(\lambda^z_{r,k,m})$ are the so-called \textit{dual variables}, and $\rho>0$ is a step-size parameter used to regulate the convergence speed of the distributed algorithm~\cite{boyd2011distributed}.

Let $k\in\clusters$, we define $\mathbf{y}_{-k} = (\mathbf{y}_{m})_{m\in\clusters\setminus\{k\}}$
which identifies the slice admission policies taken by all clusters except for cluster $k$. 
Similarly, we define $\bs{\sigma}_{-k} = (\bs{\sigma}_{m})_{m\in\clusters\setminus\{k\}}$.
Problem~\eqref{prob:desp:utility} can be solved through the following ADMM-based iterative algorithm
\begin{align}
\{\mathbf{y}_k,\bs{\sigma}_k\}(t\!+\!1) & = \argmax_{\mathbf{y}_k,\bs{\sigma}_k} L(\bs{\sigma}_k,\mathbf{y}_k,\bs{\sigma}_{-k}(t),\mathbf{y}_{-k}(t),\bs{\lambda}(t),\rho) \label{eq:admm:prob}\\
    \lambda^z_{r,k,m}(t\!+\!1) & = \lambda^z_{r,k,m}(t) + \rho (y^z_{r,k}(t+1) - y^z_{r,m}(t+1)) \label{eq:admm:lambda}
\end{align}
\noindent
where each cluster sequentially updates $\mathbf{y}_k$ and $\bs{\sigma}_k$, while the dual variables $\bs{\lambda}$ are updated as soon as all clusters have updated their strategy according to \eqref{eq:admm:prob}. 
%
To update \eqref{eq:admm:prob} each cluster solves the following quadratic problem
\begin{align} 
 \underset{\bs{\sigma}_k,\mathbf{y}_k}{\text{maximize}} & \! \sum_{z\in\mc{T}} \! \sum_{r\in\req^z} \!\tilde{v}^z_{r,k}(\mathbf{y}_{-k}(t\!-\!1),\bs{\lambda}(t\!-\!1))  y^z_{r,k} \!-\! 2\rho (y^z_{r,k})^2 \label{prob:dcesp:utility} \tag{DC-ESP} \\
 \text{subject to} & \hspace{0.2cm} \mbox{Constraints}~\eqref{eq:c2},~\eqref{eq:c4},~\eqref{eq:desp:c1},~\eqref{eq:desp:c2}, \nonumber
\end{align}
\noindent
where $\tilde{v}^z_{r,k}$ is the adjusted value of request $r$ defined as 
\begin{align}
    \tilde{v}^z_{r,k}(\mathbf{y}_{-k}(t),\bs{\lambda}(t)) & = v^z_{r,k}  - \sum_{m\in\clusters\sm{k}} \left( \lambda^z_{r,k,m}(t) - \lambda^z_{r,m,k}(t)\right) \nonumber \\
    & + \rho~ \phi_{r,k}(\mathbf{y}_{-k}(t)) \label{eq:value}
\end{align}
\noindent
and $\phi_{r,k}(\mathbf{y}_{-k}(t)) = \sum_{m\in\clusters\sm{k}} y^z_{r,m}(t)$ is used by cluster $k$ to obtain the number of clusters that have accepted request $r$. 

The advantages of Problem \eqref{prob:dcesp:utility} are that (i) clusters do not need to advertise the composition of the physical infrastructure to the IP or to other clusters, and (ii) it can be implemented in a distributed fashion. Indeed, at any iteration $t$, the only parameters needed by cluster $k$ to solve \eqref{eq:admm:prob} are the dual variables $\bs{\lambda}(t-1)$ and the number $\phi_{r,k}(\mathbf{y}_{-k}(t))$ of clusters that admitted the request $r$ at the previous iteration. 

It has been shown that ADMM usually enjoys linear convergence~\cite{shi2014linear}, but improper choices of  $\rho$ might generate oscillations. To overcome this issue and achieve convergence, we implemented the approach proposed in \cite[Eq.~(3.13)]{boyd2011distributed} where $\rho$ is updated at each iteration of the ADDM. The optimality and convergence properties of DC-ESP will be exstensively evaluated in Figs. \ref{fig:epsilon:iter} and \ref{fig:epsilon:gap}.



\section{Numerical Results} \label{sec:numerical}

We now assess the performance of the three slicing algorithms described in Section~\ref{sec:optimal} and Section~\ref{sec:approximation} by (i) simulating a MEC-enabled 5G network, and by (ii) comparing the algorithms with the recently-published DIRECT framework \cite{direct}.

We consider a scenario where edge nodes provide mobile subscribers with 5G NR connectivity as well as storage and computation MEC services, such as caching and video decoding. We assume that (i) edge nodes share the same NR numerology---more precisely, networking resources are arranged over an OFDM-based resource grid with 50~RBs, and (ii) edge nodes are equipped with hardware components with up to 1~Terabyte of storage capabilities and a maximum of 200~GIPS. We fix the number of RBs for each edge node, while we randomly generate the amount of computation and storage resources at each simulation run. To simulate a realistic scenario with video transmission, storage and transcoding applications, collateral matrices in \eqref{mat:alpha} are generated by randomly perturbing the following matrix $\mathbf{A}^0=[1, 0.0382, 0.1636; 26.178, 1, 0.0063; 0.49, 0.15, 1]$ at each run. 
To give an example, processing a data rate of $15.264$~Mbit/s (equivalent to LTE 16-QAM with 50 RBs) requires $24.4224$~GIPS (\textit{e.g.}, turbo-decoding)~\cite{holma2009lte}, which results in $A_{d,k}^{\mc{N}\rightarrow \mc{C}} = 0.49$~GIPS/RB. Similarly, a 1-second long compressed FullHD 30fps video approximately occupies $500$~kB and requires $80$~GIPS to decode, thus $A_{d,k}^{\mc{C}\rightarrow \mc{S}}=0.0062$~MB/GIPS. We assume that the physical infrastructure consists of $K=5$ MEC-enabled edge clusters, each containing the same number of edge nodes but equipped with different amount of available resources and collateral functions. We model $\Delta(\cdot)$ as the \textit{cosine similarity function} \cite{huang2008similarity} and, unless otherwise stated, the aggregation threshold is set to $\epsilon = 0.1$. Slice requests and the demanded resources are randomly generated at each run.

In the following, we refer to the optimal B\&C algorithm in Section~\ref{sec:optimal} as O-ESP. Similarly, the two approximation algorithms proposed in Section~\ref{sec:aggreagation} and Section~\ref{sec:distributed} are referred to as V-ESP and DC-ESP respectively.

\subsection{The impact of coupling on MEC-enabled 5G systems} \label{sec:num:over-provisioning}

DIRECT \cite{direct} provides an efficient distributed slicing algorithm for networking and computing resources in MEC-enabled 5G networks. 
Although this approach does not account for the case where edge nodes provide both networking and MEC functionalities, it represents the closest work to ours. 

\begin{figure}[t]
    \centering
    \setlength\belowcaptionskip{-.5cm}
    \includegraphics[width=\columnwidth]{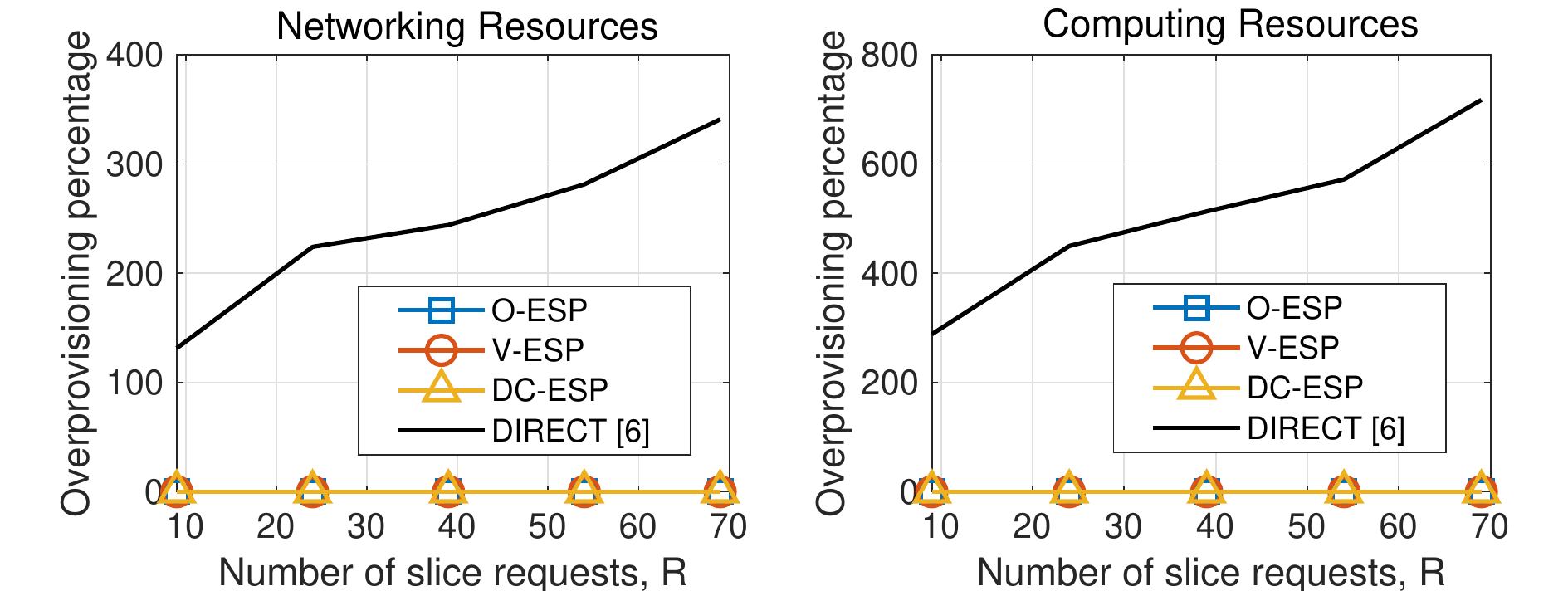}
    \caption{Over-provisioning of networking and computational resources of \textit{Sl-EDGE} and DIRECT \cite{direct}.}
    \label{fig:over-provisioning}
 \end{figure}

Moreover, DIRECT does not explicitly slice storage resources. Thus, to perform a fair comparison, we consider the case where tenants do not request any storage resource. Let $D_c = 75$ be the total number of edge nodes in the network. We let tenants randomly generate slice requests to obtain networking and computational resources. Results are shown in \reffig{fig:over-provisioning}, where any positive value indicates resource over-provisioning.

\reffig{fig:over-provisioning} shows that \textit{Sl-EDGE} never produces over-provisioning slices. Conversely, since DIRECT does not account for coupling among heterogeneous resources on the same edge node, it always incurs in over-provisioning, allocating up to 6x more resources than the available ones. These results conclude that \textit{already existing solutions, which perform well in 5G systems with networking and MEC functionalities decoupled at different points of the network, cannot be readily applied to scenarios where resources are simultaneously handled by edge nodes---which strongly motivates the need for approaches such as \textit{Sl-EDGE}}.

\subsection{Maximizing the number of admitted slices} \label{sec:num:number}

Let us now focus on the scenario where the IP owning the infrastructure aims at maximizing the number of slice requests admitted by \textit{Sl-EDGE}---to maximize resource utilization, for instance.
Although each slice request $r$ comes with an associated (monetary) value $v_r>0$, the above can be achieved by resetting the value of each request to $v_r = 1$ in Problem \eqref{prob:utility}.

 \begin{figure}[t]
    \centering
    \setlength\belowcaptionskip{-.5cm}
    \includegraphics[width=\columnwidth]{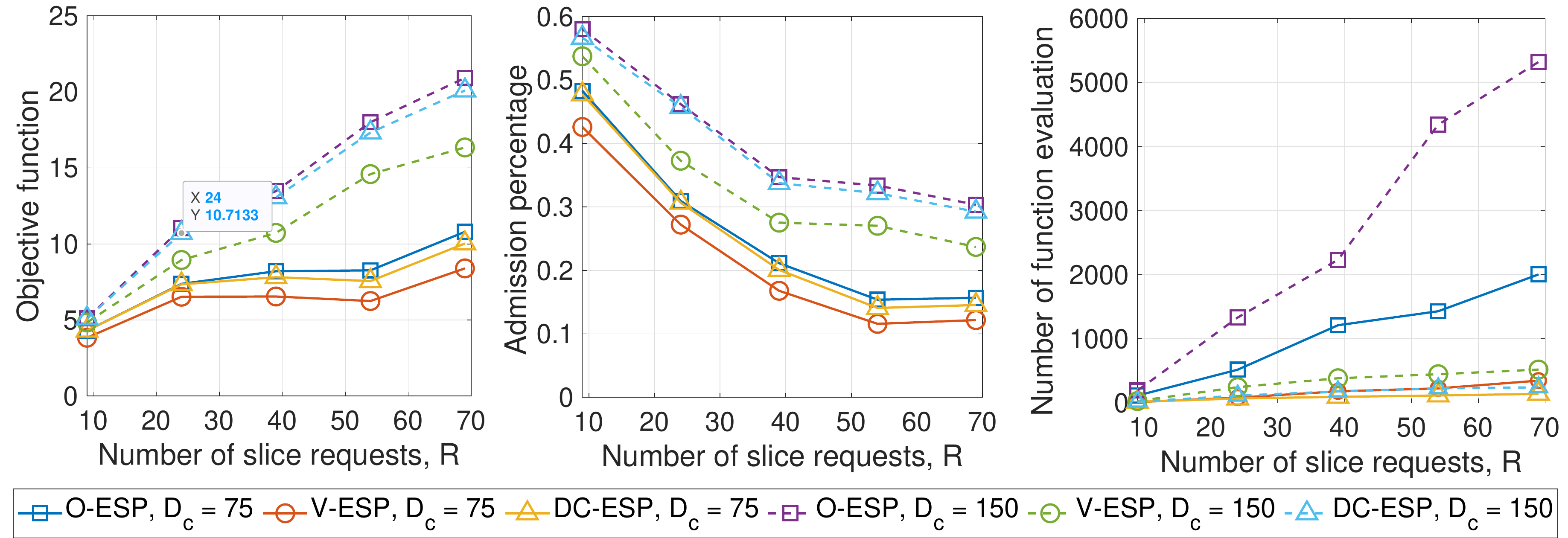}
    \caption{\textit{Sl-EDGE} performance when maximizing the number of admitted slice requests. \label{fig:r:number}}
 \end{figure}

\reffig{fig:r:number} reports the performance of \textit{Sl-EDGE} as a function of the total number $R$ of generated slice requests for different values of the number of edge nodes $D_c$. We notice that the number of admitted slices increase as the slice requests that are submitted to \textit{Sl-EDGE} increase (left-side plot). However, \reffig{fig:r:number} (center) clearly shows that the percentage of admitted slices rapidly decreases as $R$ increases (only $10$ requests are admitted by O-ESP when $D_c=75$ and $R=70$). This is due to the scarcity of resources at edge nodes, which prevents the admission of a large number of slices. Thus, IPs should either provide edge nodes with more resources, or increase the number of deployed edge nodes. \reffig{fig:r:number} (left), indeed, shows that denser deployments of edge nodes (\textit{i.e.}, $D_c=150$) allows more slices to coexist on the same infrastructure. 

Finally, the right-hand side plot of \reffig{fig:r:number} shows the computational complexity of the three algorithms measured as the number of function evaluations needed to output a solution. As expected, the complexity of all algorithms increases as both $R$ and $D_c$ increase. Moreover, O-ESP, a fully centralized algorithm, has the highest computational complexity. V-ESP and DC-ESP, reduced-complexity versions of O-ESP, instead show lower complexity. However, V-ESP and DC-ESP admit approximately $10\%$ and $16\%$ less requests than O-ESP, respectively, due to their non-optimality.

\subsection{Maximizing the profit of the IP} \label{sec:num:value}

Let us now consider the case of slice admission and instantiation for profit maximization  (\reffig{fig:r:value}). In this case, \textit{Sl-EDGE} selects the slice requests to be admitted to maximize the total (monetary) value of the admitted slices. Similarly to \reffig{fig:r:number}, \reffig{fig:r:value} (center) shows that increasing $R$ reduces the percentage of admitted slices.

%
 \begin{figure}[!t]
    \centering
    \setlength\belowcaptionskip{-.3cm}
    \includegraphics[width=\columnwidth]{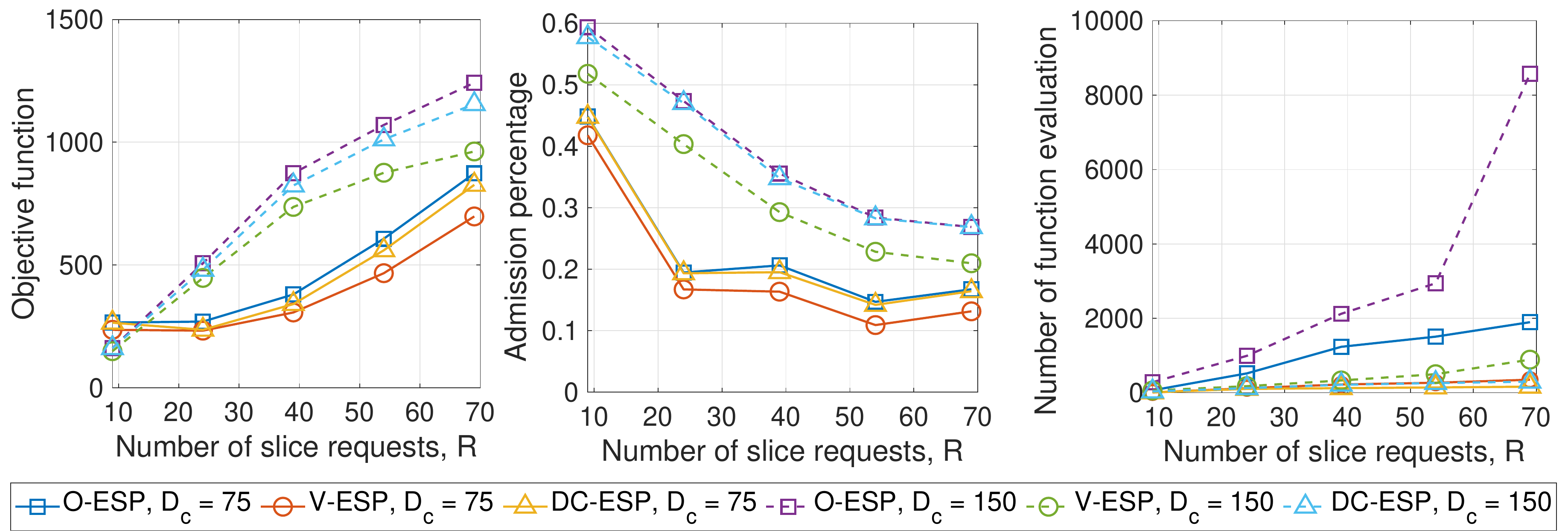}
    \caption{\textit{Sl-EDGE} performance maximizing the profit of the IP. \label{fig:r:value}}
 \end{figure}
\begin{figure}[t]
    \setlength\abovecaptionskip{-0.1cm}
    \setlength\belowcaptionskip{-0.4cm}
    \centering
    \includegraphics[width=0.8\columnwidth]{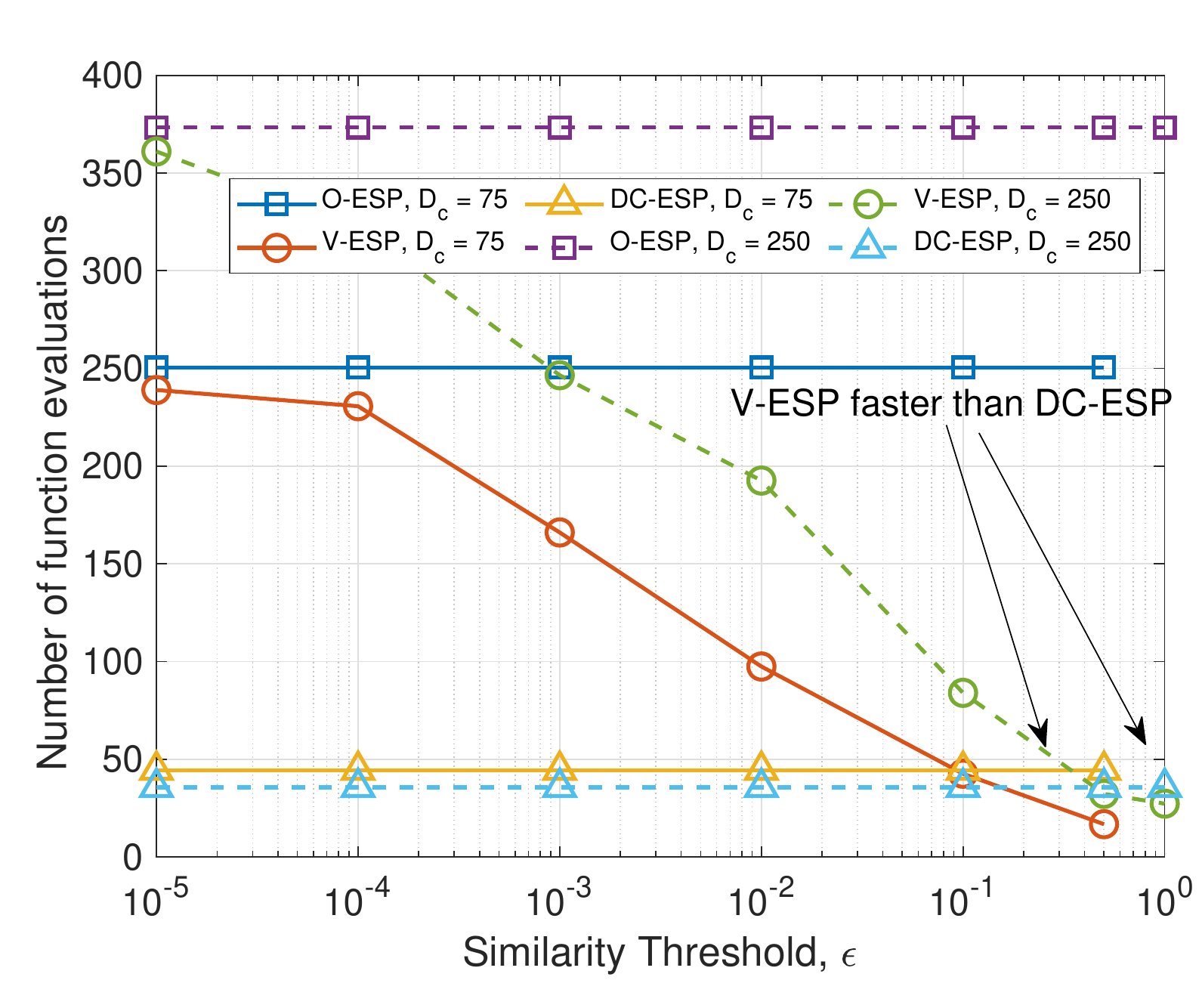}
    \caption{Computational complexity of the proposed algorithms as a function of the similarity parameter $\epsilon$.}
    \label{fig:epsilon:iter}
 \end{figure}
 
When compared to the problem described in Section~\ref{sec:num:number}, this profit maximization problem differs because (i) even if the number of edge nodes is small (\textit{i.e.,} $D_c = 75$), profit maximization produces profits that rapidly increase with $R$, and (ii) the percentage of admitted requests steeply decreases as $R$ increases. Indeed, the higher the number of requests, the higher the probability that slices with high value are submitted by tenants. In this case, \textit{Sl-EDGE} prioritizes more valuable requests at the expenses of others.

\subsection{Impact of $\epsilon$ on the V-ESP algorithm} \label{sec:numerical:epsilon}

Finally, we investigate the impact of different choices on the performance of the V-ESP algorithm. Recall that $\epsilon$ regulates the number of edge nodes that are aggregated into virtual edge nodes (Section \ref{sec:aggreagation}). The higher the value of $\epsilon$, the higher the percentage of edge nodes that are aggregated, and the smaller the number of virtual edge nodes generated by \textit{Sl-EDGE}.
 
 \reffig{fig:epsilon:iter} shows the computational complexity of V-ESP as a function of $\epsilon$ for different number $D_c$ of deployed edge nodes. As expected, $\epsilon$ does not impact neither O-ESP nor DC-ESP, however the impact on V-ESP is substantial. Indeed, larger values of $\epsilon$ reduce the number of physical edge nodes in the network, which are instead substituted by virtual edge nodes (one per aggregated group). This reduction eventually results in decreased computational complexity. Surprisingly, \reffig{fig:epsilon:iter} also shows that V-ESP enjoys an even lower computational complexity than that of the distributed DC-ESP when $\epsilon \approx 1$. Recall that V-ESP centrally determines an efficient slicing strategy over virtualized edge nodes, and these strategies are successively enforced by each cluster. This means that V-ESP can compute an efficient slicing policy as rapidly as DC-ESP while avoiding any coordination among different clusters. Overall, \reffig{fig:epsilon:iter} shows that V-ESP (green dashed line) computes a solution 7.5x faster than O-ESP (purple dashed line) when $\epsilon$ is high.

\begin{figure}[t]
    \setlength\belowcaptionskip{-.4cm}
    \centering
    \vspace{-0.5cm}
    \includegraphics[width=0.9\columnwidth]{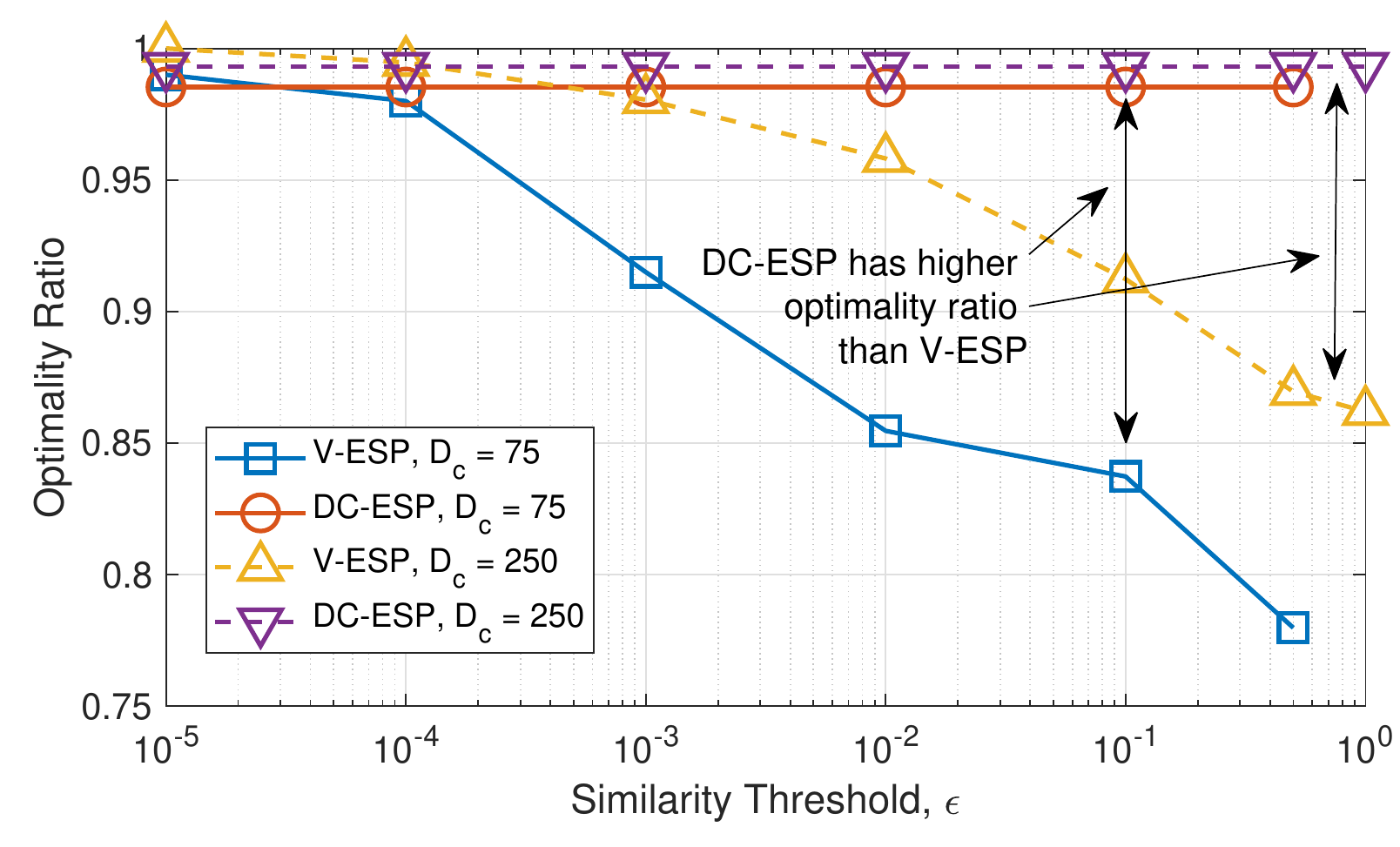}
    \caption{Optimality ratio of the algorithms proposed in Section \ref{sec:approximation} as a function of the similarity parameter $\epsilon$.}
    \label{fig:epsilon:gap}
 \end{figure}
 
 However, \reffig{fig:epsilon:gap} shows that reduced computation complexity comes at the expense of efficiency. 
 \begin{figure}[b]
    \centering
    \includegraphics[width=\columnwidth]{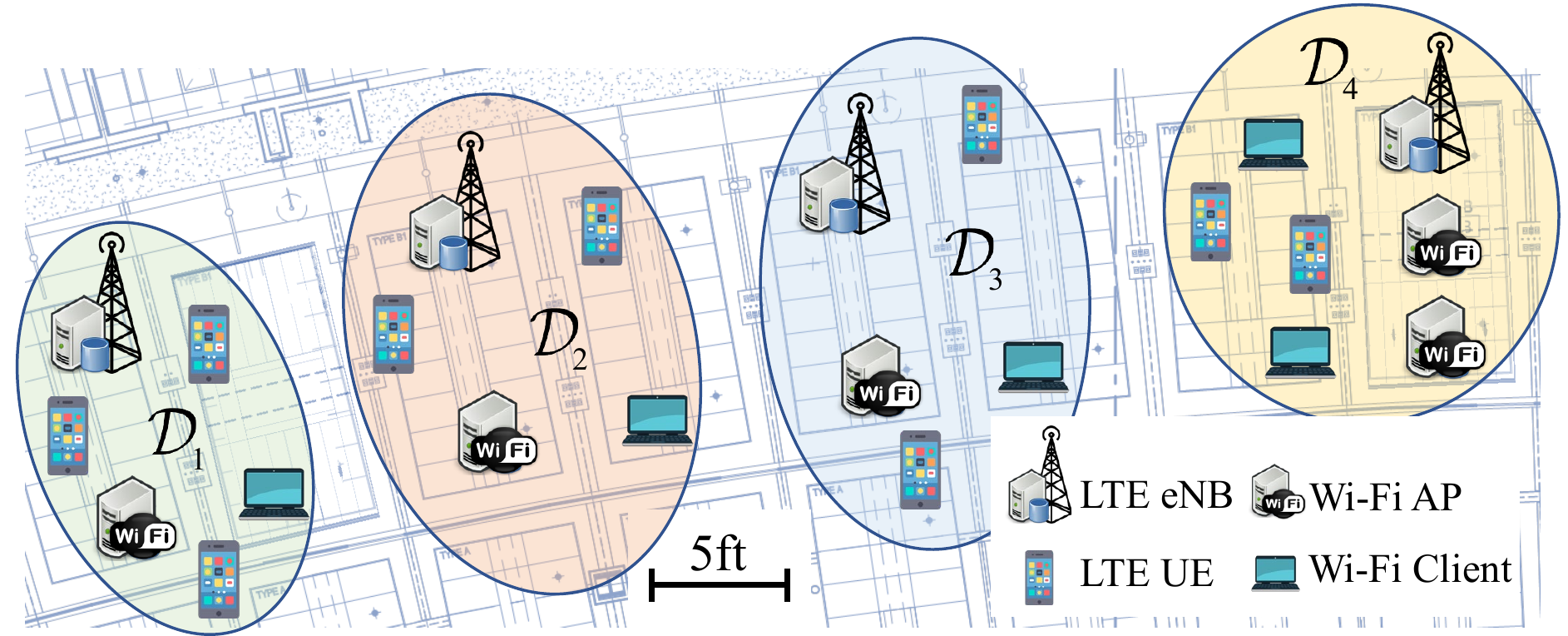}
    \caption{\textit{Sl-EDGE} testbed configuration.}
    \label{fig:testbed}
\end{figure}
 Indeed, the optimality ratio (\textit{i.e.}, the distance of the output of any approximation algorithm from the optimal solution of the problem) decreases as $\epsilon$ increases up to a maximum of 25\% loss with respect to the optimal. Although the optimality ratio for $\epsilon = 0.1$ is high (\textit{i.e.,} $92\%$ and $84\%$ for $D_c = 75$ and $D_c=250$, respectively), clearly a trade-off between computational complexity and efficacy should be considered. 

\section{S\lowercase{l}-EDGE prototype} \label{sec:experimental}

We prototyped \textit{Sl-EDGE} on \textit{Arena}, a large-scale 64-antennas SDR testbed~\cite{bertizzolo2019arena}. A server rack composed of 12 Dell PowerEdge R340 servers is used to control the testbed SDRs, and to perform baseband signal processing as well as generic computation and storage operations.
The servers connect to a radio rack formed of 24 Ettus Research SDRs (16 USRPs N210 and 8 USRPs X310) through $10\:\mathrm{Gbit/s}$ optical fiber cables to enable low-latency and reactive communication with the radios.
These are connected to 64 omnidirectional antennas through $100\:\mathrm{ft}$ coaxial cables.
Antennas are hung off the ceiling of a $2240\:\mathrm{ft^2}$ office space and operate in the 2.4-2.5 and 4.9-5.9~GHz frequency bands.
The USRPs in the radio rack achieve symbol-level synchronization through four OctoClock clock distributors.

We leveraged 14 USRPs of the above-mentioned testbed (10 USRPs N210 and 4 USRPs X310) to prototype \textit{Sl-EDGE}. In our testbed, an edge node consists of one USRP and one server, the former provides networking capabilities, while the latter provides storage and computing resources. 
The testbed configuration adopted to prototype and evaluate \textit{Sl-EDGE} performance is shown in \reffig{fig:testbed}.

Since there are no open-source experimental 5G implementations yet, we used the LTE-compliant \textit{srsLTE} software~\cite{gomez2016srslte} to implement LTE \textit{networking} slices. Since LTE and NR resource block grids are  similar, we are confident that our findings remain valid for 5G scenarios. Specifically, srsLTE offers a standard-compliant implementation of the LTE protocol stack, including Evolved Packet Core and LTE base station (eNB) applications. We leveraged srsLTE to instantiate 4 eNBs on USRPs X310, while we employed 9 COTS cellular phones (Samsung Galaxy S5) as users. Each user downloads a data file from one of the rack servers, which are used as caching nodes with \textit{storage} capabilities.

We consider three tenants demanding an equal number of LTE network slices (\textit{i.e.,} $\mbox{LS}_1$, $\mbox{LS}_2$ and $\mbox{LS}_3$) at times $t_0 = 0\:\mathrm{s}$, $t_1 = 40\:\mathrm{s}$, and $t_2 = 80\:\mathrm{s}$. Each tenant controls a single slice only and serves a set of UEs located in different clusters as shown in \reftable{tab:lte:slice1} (right).

\begin{table}[t]
\centering
\caption{Per-cluster admitted RBs in $\mbox{LS}_1$, and UE association.}
\label{tab:lte:slice1}
\small
\setlength{\tabcolsep}{3.5pt}
\renewcommand{\arraystretch}{1.2}
\begin{tabular}[]{c|c|c|c|c|c|c|c|}
\cline{2-4}\cline{6-8}
 & $t_0\!=\![0, 40]\mathrm{s}$ & $t_1\!=\![40, 80]\mathrm{s}$ & $t_2\!=\![80, 160]\mathrm{s}$ & & $\mbox{LS}_1$   & $\mbox{LS}_2$   & $\mbox{LS}_3$ \\
\cline{1-4}\cline{6-8}
\multicolumn{1}{|c|}{$\mathcal{D}_1$} & 24      & 0     &  0 & &  $\mbox{UE}_1$ & $\mbox{UE}_2$ & $\mbox{UE}_7$\\
\cline{1-4}\cline{6-8}
\multicolumn{1}{|c|}{$\mathcal{D}_2$} & 0       & 0     & 0 & & - & $\mbox{UE}_3$ & $\mbox{UE}_8$ \\
\cline{1-4}\cline{6-8}
\multicolumn{1}{|c|}{$\mathcal{D}_3$} & 24      & 24    & 0 & & $\mbox{UE}_4$ & $\mbox{UE}_5$ & $\mbox{UE}_9$ \\
\cline{1-4}\cline{6-8}
\multicolumn{1}{|c|}{$\mathcal{D}_4$} & 42      & 24    & 0 & & $\mbox{UE}_6$ & - & - \\
\cline{1-4}\cline{6-8}
\end{tabular}
\end{table}

To test \textit{Sl-EDGE} abilities in handling slices involving both \textit{networking and computation capabilities}, we also implemented a video streaming slice where edge nodes stream a video file stored on an Apache instance through the \textit{dash.js} reference player~\cite{dashjs} running on the Chrome web browser. DASH allows real-time adaptation of the video bitrate, according to the client requests and the available resources~\cite{dash}. Each streaming video was sent to the receiving server of the rack through USRPs N210 acting as SDR-based WiFi transceivers (WiFi Access Points~(APs) and Clients in \reffig{fig:testbed}), using the GNU Radio-based IEEE~802.11a/g/p  implementation~\cite{bloessl2018performance}.
In the meanwhile, the edge node performed transcoding of video chunks using \textit{ffmpeg}~\cite{ffmpeg}. We point out that each SDR can play multiple roles in the cluster (\textit{e.g.,} USRPs X310 can act as WiFi transceiver/LTE eNB), and the actual role is determined at run-time based on the slice types allocated to each tenant.

A demonstration of the operations of \textit{Sl-EDGE} in the scenario of~\reffig{fig:testbed} is shown in Figs.~\ref{fig:exp:lte},~\ref{fig:exp:videoBR} and~\ref{fig:exp:cpu}. Overall, the prototype of \textit{Sl-EDGE} allocates and supports 11 heterogeneous slices simultaneously: 3 for cellular connectivity, 3 for video streaming over WiFi, and 5 for computation with the \textit{ffmpeg} transcoding. Bitrate results for the LTE slices and individual UEs are reported in \reffig{fig:exp:lte}, where we show that \textit{Sl-EDGE} provides an overall instantaneous throughput of $37~\mathrm{Mbit/s}$. 

\begin{figure}[h]
     \centering
     \setlength\belowcaptionskip{-.2cm}
     \subfloat{\includegraphics[width=0.33\columnwidth]{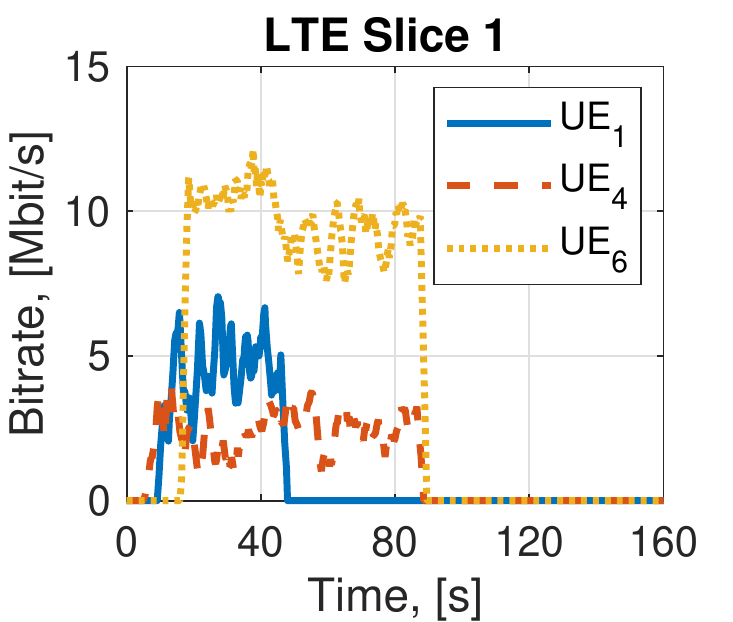}\label{fig:exp:s1}}
     \subfloat{\includegraphics[width=0.33\columnwidth]{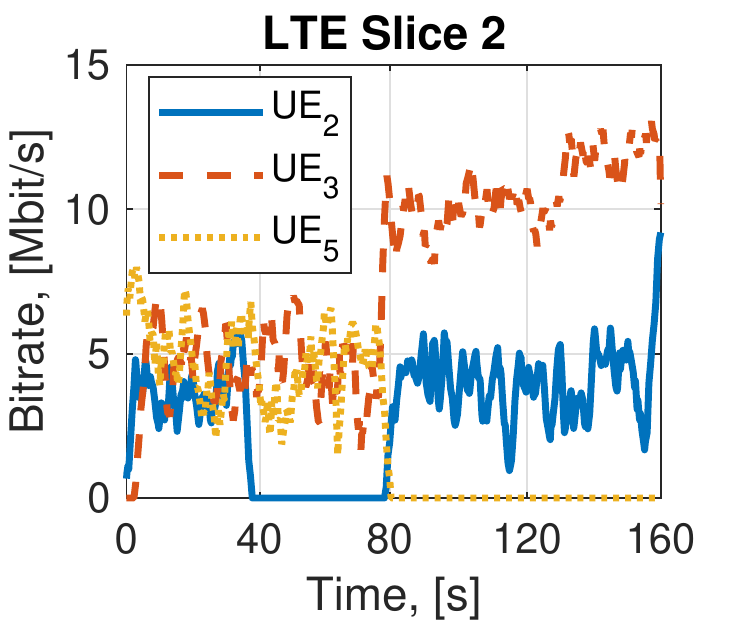}\label{fig:exp:s2}}
     \subfloat{\includegraphics[width=0.33\columnwidth]{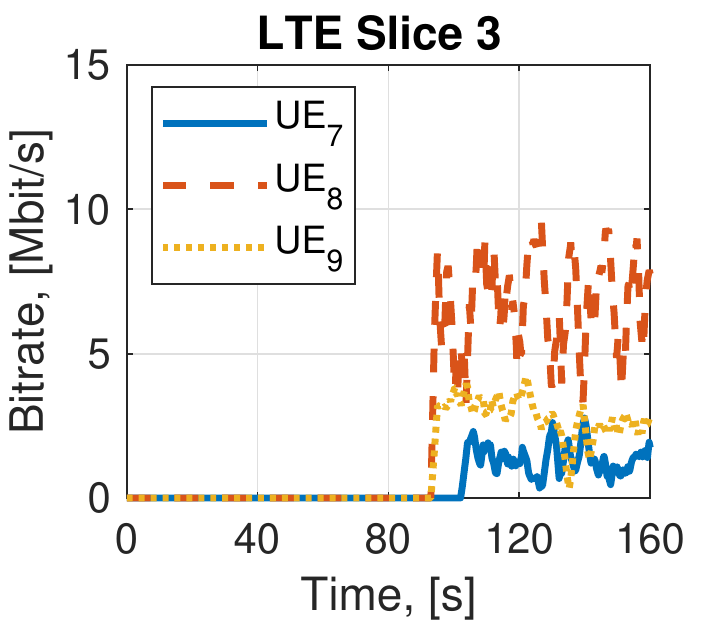}\label{fig:exp:s3}}
     \vspace{-0.2cm}
     \caption{Instantiation of LTE network slices.}
     \label{fig:exp:lte}
\end{figure}

It is worth to point out that the throughput of each LTE slice, and thus each UE, varies according to the amount of resources allocated to the tenants. An example is shown in \reftable{tab:lte:slice1} (left), where we report the output of \textit{Sl-EDGE} O-ESP algorithm (\textit{i.e.}, the amount of RBs allocated to LTE Slice~1 ($\mbox{LS}_1$)) in each cluster. 
Such an allocation impacts the throughput of UEs attached to slice $\mbox{LS}_1$. As an example, in \reffig{fig:exp:lte} we notice that $\mathrm{UE}_6 \in \devs_4$ is allocated 42~RBs at $t_0$, 24~RBs in $t_1$, and 0~RBs in $t_2$ and approximately achieves a throughput of $12\mathrm{Mbit/s}$, $8\mathrm{Mbit/s}$ and $0\mathrm{Mbit/s}$, respectively.

The video streaming application from~\reffig{fig:exp:videoBR} involves 5 tenants that share 3 non-overlapping channels, allocated in any cluster $\devs_i, i \in \{1, 2, 3, 4\}$. To avoid co-channel interference, \textit{Sl-EDGE} only admits 3 flows at any given time. As Fig.~\ref{fig:exp:videoBR} shows, during the first 70 seconds of the experiment only the slices for tenants 1, 2 and 3 are admitted, while tenant 4 needs to wait for tenant 3 to stop the video streaming before being granted a slice. Similarly, the slice for tenant 5 starts at time $t = 140\:\mathrm{s}$, when the flow of tenant 1 stops. Meanwhile, the tenants submit requests for computation slices to transcode the videos with ffmpeg, which compete with srsLTE and GNU Radio slices necessary for LTE connectivity and video streaming in the 3 LTE eNBs and 5 WiFi APs of the 4 clusters. 
Moreover. in each server one of the 6 cores is reserved to the operating system exclusively and is never allocated to tenants.  Fig.~\ref{fig:exp:cpu} shows that \textit{Sl-EDGE} limits the total CPU utilization to 83\%, which demonstrates that \textit{Sl-EDGE} avoids over-provisioning of available resources.
 
\begin{figure}[t!]
    \centering
    \setlength\belowcaptionskip{-.2cm}
    \includegraphics[width=\columnwidth]{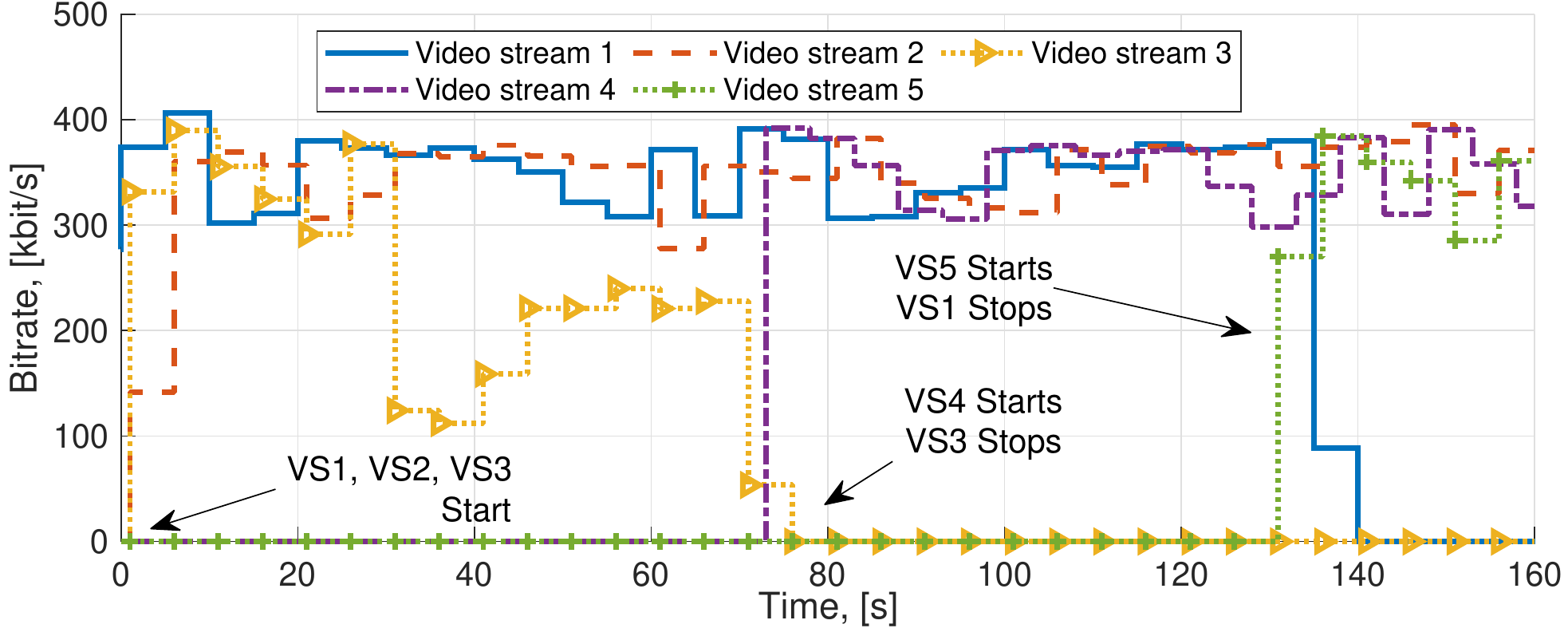}
    \vspace{-0.4cm}
    \caption{Dynamic instantiation of video streaming slices.}
    \label{fig:exp:videoBR}
\end{figure}

\begin{figure}[t!]
    \centering
    \setlength\belowcaptionskip{-.2cm}
    \includegraphics[width=\columnwidth]{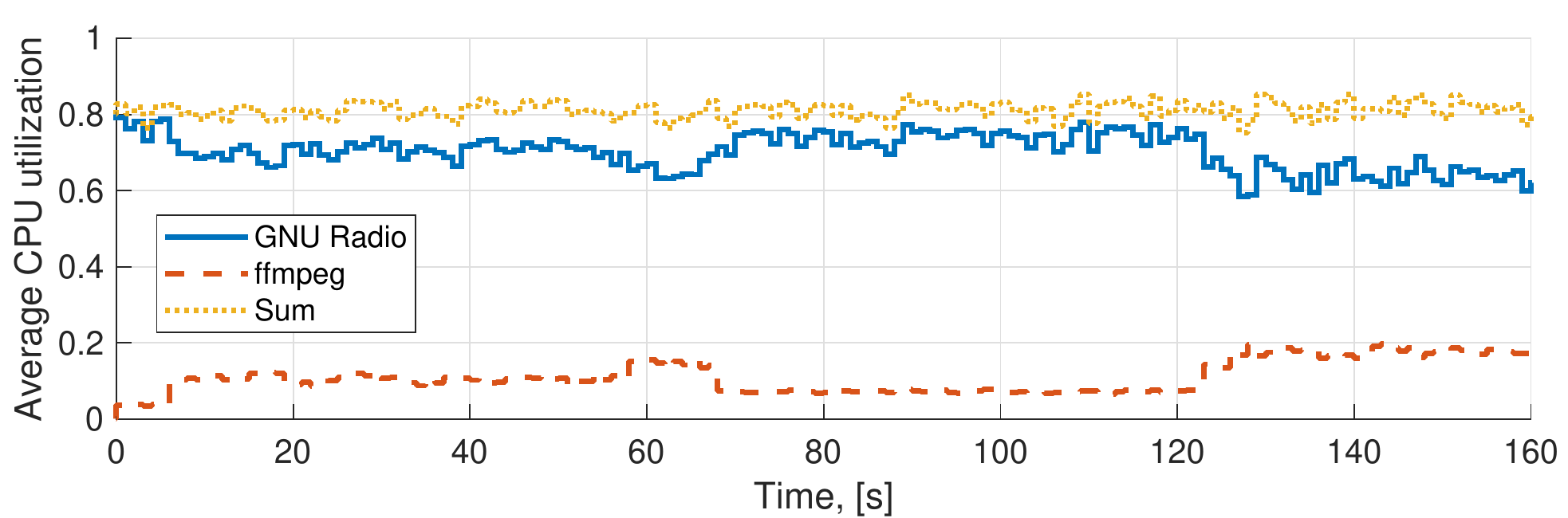}    
    \vspace{-0.4cm}
    \caption{CPU utilization for networking and transcoding services. }
    \label{fig:exp:cpu}
\end{figure}

\section{Conclusions} \label{sec:conclusions}
In this paper, we have presented \textit{Sl-EDGE}, a unified MEC slicing framework that instantiates heterogeneous slice services (\textit{e.g.,} video streaming, caching, 5G network access) on edge devices. We have first shown that the problem of optimally instantiating joint network-MEC slices is NP-hard. Then, we have proposed distributed algorithms that leverage similarities among edge nodes and resource virtualization to instantiate heterogeneous slices faster and within a short distance from the optimum. We have assessed the performance of our algorithms through extensive numerical analysis and on a 64-antenna testbed with~24 software-defined radios. Results have shown that \textit{Sl-EDGE} instantiates slices 6x more efficiently then state-of-the-art MEC slicing algorithms, and that \textit{Sl-EDGE} provides at once highly-efficient slicing of joint LTE connectivity, video streaming over WiFi, and ffmpeg video transcoding.

\begin{acks}
This work was supported in part by ONR under Grants N00014-19-1-2409 and N00014-20-1-2132, and by NSF under Grant CNS-1618727.
\end{acks}

\balance
\bibliographystyle{acm}
\bibliography{bibl} 

\end{document}